\documentclass[journal,10pt]{IEEEtran}

\usepackage{amsfonts}
\usepackage{mathrsfs}
\usepackage{amsfonts}
\usepackage{amssymb}
\usepackage{graphicx}
\usepackage{epsfig}
\usepackage{psfrag}
\usepackage{amsmath}
\usepackage{array}
\usepackage{cases}
\usepackage{algorithm2e}
\usepackage{subfigure}
\usepackage{cite,graphicx,amsmath,amssymb,color}
\usepackage{anysize}

\usepackage{algorithmic}
\usepackage{stmaryrd}
\usepackage{multirow}
\usepackage{subfig}
\usepackage{graphicx,times,amsmath} 

\marginsize{13mm}{13mm}{19mm}{43mm}

\begin{document}

\allowdisplaybreaks

\title{Enhancing  the Delay Performance of \\Dynamic Backpressure Algorithms}
\author{\authorblockN{Ying Cui,~\IEEEmembership{Member,~IEEE}, \ Edmund M. Yeh,~\IEEEmembership{Senior Member,~IEEE} and Ran Liu,~\IEEEmembership{Student Member,~IEEE}\thanks{Manuscript received March 26, 2014; revised December  19, 2014; accepted January 12, 2015; approved by IEEE/ACM TRANSACTIONS ON NETWORKING Editor L. Ying.  
Y. Cui  was supported by the National Science Foundation of China grant 61401272.  E. Yeh was supported by the National
Science Foundation grant CNS-1205562.
This paper was presented in part at IEEE Asilomar Conference on Signals, Systems and Computers, November 2013,  and at 
IEEE International Symposium on Information Theory (ISIT), July 2014.}\thanks{
Ying Cui is with the Department of  Electronic Engineering, Shanghai Jiao Tong University, China (e-mail: cuiying@sjtu.edu.cn).} \thanks{Edmund Yeh and Ran Liu are with the Department of Electrical  and Computer Engineering, Northeastern University, USA (e-mail: eyeh@ece.neu.edu; liu.ran1@husky.neu.edu).}
}}
\maketitle

\newtheorem{Thm}{Theorem}
\newtheorem{Lem}{Lemma}
\newtheorem{Cor}{Corollary}
\newtheorem{Def}{Definition}
\newtheorem{Exam}{Example}
\newtheorem{Alg}{Algorithm}
\newtheorem{Sch}{Scheme}
\newtheorem{Prob}{Problem}
\newtheorem{Rem}{Remark}
\newtheorem{Proof}{Proof}
\newtheorem{Asump}{Assumption}

\begin{abstract}
For general multi-hop queueing networks,  delay optimal network control has  unfortunately been an outstanding problem. The dynamic backpressure (BP) algorithm elegantly achieves throughput optimality, but does not yield good delay performance in general.  In this paper, we  obtain an asymptotically delay optimal control policy, which resembles the BP algorithm in basing resource allocation and routing on a backpressure calculation, but differs from the BP algorithm in the form of the backpressure calculation employed.  
The difference suggests a possible reason for the unsatisfactory delay performance of the BP algorithm, i.e., the myopic nature of the BP control. 
Motivated by this new connection, we introduce a new class of enhanced backpressure-based algorithms which incorporate  a general queue-dependent bias function into the backpressure term of the traditional BP algorithm to improve  delay performance. These enhanced algorithms exploit  queue state information beyond one hop.  We prove the throughput optimality and characterize the utility-delay tradeoff of the enhanced algorithms.  We further focus on two specific distributed algorithms within this class, which have demonstrably improved delay performance as well as acceptable implementation complexity. 
 \end{abstract}

\begin{IEEEkeywords}
dynamic backpressure algorithms, congestion control, delay optimal control, throughput optimal control, dynamic programming, Lyapunov drift.
\end{IEEEkeywords}

\section{Introduction}

With the significant increase in demand for real-time services, it is well recognized that  networks must be jointly optimized across the physical, medium access control (MAC), and network layers to support delay-sensitive applications. 
{\em Delay optimal}  network control for general multi-hop queueing networks, which seeks to minimize some function of average delay (or average queue size) by incorporating resource allocation and routing across different layers, has unfortunately been an outstanding problem for some time.  Often, even the basic structural properties of the delay optimal control policy are not known.   While dynamic programming represents a systematic approach for delay optimal control, there generally exist only numerical solutions\cite{CuiDelayTutorial:2012,CuiOFDMAUL:2010,Cuidynclustering2011,CuiDTC2012,RayRS-MIMO:2011}.  These solutions do not typically offer many design insights and are usually impractical for implementation in large-scale multi-hop networks, due to the curse of dimensionality \cite{Bertsekas:2007}.

A notable success in networking research is the formulation of the {\em throughput optimal} network control problem and its solution via the {\em dynamic backpressure (BP) algorithm} \cite{Tassiulas-Ephremides:1992-2,Georgiadis-Neely-Tassiulas:2006}. Throughput optimal control seeks to ensure the stability of  general multi-hop queueing networks (all queue sizes remain finite for all time) for any arrival rate vector within the network stability region.  The BP algorithm is obtained  using Lyapunov drift techniques. It incorporates resource allocation and routing across the physical, MAC, and network layers, and elegantly achieves throughput optimality via load balancing \cite{Tassiulas-Ephremides:1992-2, Georgiadis-Neely-Tassiulas:2006}.  The algorithm has also been combined with flow control in the transport layer to yield maximum network utility when the data arrival rate is outside the network stability region~\cite{Georgiadis-Neely-Tassiulas:2006}. 
One major shortcoming of the BP algorithm, however, is that it does not yield good delay performance in general.  
In routing packets, the BP algorithm typically explores all possible paths between sources and destinations (i.e., load balancing over the entire network), without explicitly considering delay performance. This extensive exploration is essential for maintaining stability when the network is heavily loaded. Under light or moderate loading, however,  packets may be sent over unnecessarily long routes, which leads to excessive delays.

For any arrival rate vector within the network stability region,  the delay optimal control minimizes  average  delay (average queue size), while the BP algorithm ensures finite average queue size and typically has good delay performance only under heavy load.  Therefore, two interesting questions are: (1) whether there is any subtle connection between the two network control solutions and (2) what accounts for the delay performance gap between them. 
Better understanding of these two questions may motivate the design of  enhanced BP algorithms with improved delay performance.  There are several potential challenges toward this direction. First, it is not clear how one should  improve the delay performance of the BP algorithm by  approximating the  delay optimal control in a tractable manner, in order to avoid the prohibitively high complexity of dynamic programming. Second, it is not clear how to maintain the desirable throughput optimality of the BP algorithm when the BP control structure is modified for improving the delay performance. In this paper, we shall address the above questions and challenges.

\subsection{Main Contributions}

We first study the connection  between delay optimal network control and the BP algorithm (throughput optimal network control).  Using dynamic programming  and Taylor's theorem, we obtain an asymptotically delay optimal control policy when the scheduling slot duration is small.  Surprisingly, we show that the asymptotically delay optimal control, obtained using dynamic programming, shares striking similarities with the 
BP algorithm, obtained using Lyapunov drift techniques.  Specifically, the two algorithms both base resource allocation and routing on a backpressure calculation, but differ in the form of the backpressure calculation employed.  In the BP algorithm, the backpressure of a link is derived from the differences of queue lengths at the two end nodes of the link.  Thus, the BP backpressure term reflects {\em local} queue state information (QSI).  In the asymptotically delay optimal control algorithm, the backpressure of a link is derived from the differences of the derivatives of the value function of the dynamic program at the two end nodes of the link.  
Since in general the value function depends on the  global QSI, the backpressure term for the asymptotically delay optimal control is a function of  the {\em global} QSI. This observation  suggests a possible reason for the poor delay performance of the BP algorithm, i.e., the {\em myopic} nature of the control, which relies only on one-hop queue size differences. 
 To the best of our knowledge, this is the first work which provides an analytical connection between the two network control solutions.

Motivated by the above connection, we  design enhanced BP algorithms with improved delay performance via the use of QSI beyond one hop. Specifically, we present a new class of {\em enhanced BP algorithms} which maintain a generalized notion of throughput optimality while exhibiting significantly improved delay performance, relative to the traditional BP algorithm. 
In lightly or moderately loaded networks, where the delay performance of the traditional BP algorithm is poor, the enhanced BP algorithms reduce average delay by (1) exploiting the margin between the arrival rate vector and the boundary of the network stability region, and (2) making use of QSI beyond one hop in a simple and flexible manner, via the incorporation of a QSI-dependent bias function into the backpressure calculation.  We propose two specific algorithms, named BPnxt and BPmin, within this class of enhanced BP algorithms. 
These two algorithms  promise to improve delay performance by using downstream QSI to clarify congestion patterns, while allowing for distributed implementation with manageable complexity.   BPnxt  has the same implementation complexity (in order) as the traditional BP algorithm.    BPmin  has an implementation complexity which is higher (in order) than that for the traditional BP algorithm but lower than that for other BP-based control algorithms with similar delay performance. 
 Next, the delay performance of both BPnxt and BPmin can be improved  further by incorporating an extra QSI-independent shortest path bias term into the backpressure calculation.  Finally, we present a new class of  {\em enhanced joint flow control and BP algorithms} for the case where the traffic arrival rate is outside the network stability region, and demonstrate their superior utility-delay performance tradeoff.

\subsection{Related Work}

A number of previous papers have focused on improving the delay performance of the traditional BP-based algorithms.  References \cite{Neely-Modiano-Rohrs:2005} and \cite{YingShakkottai08oncombSPDBPJNET} improve the delay by incorporating the shortest path (in terms of the number of hops) concept to avoid the extensive exploration of paths in the BP algorithm.   Specifically, in \cite{Neely-Modiano-Rohrs:2005},  
a  (constant) shortest path bias, parameterized by a per-link cost $B$, is added 
to the backpressure term so that nodes are inclined to route packets toward their destinations using shorter paths. The algorithm proposed in \cite{Neely-Modiano-Rohrs:2005} is called BPbias here.  
In
\cite{YingShakkottai08oncombSPDBPJNET}, a {joint traffic-splitting and
shortest-path-aided BP routing} algorithm, called BPSP here, is proposed, where the traffic splitting is parameterized by $K$.  A hop-queue structure is used. 
The algorithm incorporates the shortest path concept
by minimizing the average number of hops between sources and destinations,  using the hop-queue length difference in the backpressure term.  The traditional BP and BPbias algorithms require  $\mathcal O(N^2C)$ computational complexity for the backpressure calculation in each slot, where $N$ and $C$ are the number of nodes and the number of commodities in the network, respectively.   The BPSP algorithm, on the other hand, requires  $\mathcal O(N^4C)$ computational complexity for the backpressure calculation in each slot. 
As shown in Section~\ref{sec:simu}, one potential challenge for BPbias and BPSP is that their delay performance relies heavily on the choices for parameters $B$ and $K$. $B$ and $K$ must be selected for particular levels of traffic loading, which may be difficult to predict  beforehand in practical networks.  

 Reference \cite{AlresainiINFOCOM2012} improves the delay  of the traditional BP algorithm by introducing  redundant traffic and  a duplicate queue structure with finite buffers, to avoid delay increase  due to low queue occupancy.  The one-hop (finite and local) duplicate queue length difference is added to the  backpressure term, thereby capturing limited congestion information in the network. In \cite{BuiSrikant:2011}, a shadow queueing architecture is proposed  to improve the delay performance of the traditional BP algorithm by reducing the end-to-end queue length difference.  The one-hop  shadow queue length difference is used as the  backpressure term.  In \cite{ShroffdelaybasedDBP:2011}, a delay-based BP algorithm is proposed to improve the delay performance   using a new delay metric. The proposed algorithms in \cite{BuiSrikant:2011} and \cite{ShroffdelaybasedDBP:2011}, however, are designed for networks in which the routes of each flow are fixed before the arrival of packets. Reference  \cite{NeelyLMDBP11} 
improves the order of the utility-delay tradeoff  by forming a virtual backlog process and using the LIFO service discipline on top of the traditional BP-based algorithm.  The order improvement, however, hinges on the ability to drop a certain fraction of packets. In \cite{JavidiORCD1:2010,Neely09DIVBAR}, receiver diversity is used to improve the delay performance of the traditional BP algorithm for  networks with unreliable channel conditions. 

The motivation, method and design for our proposed class of enhanced BP algorithms are novel and differ significantly from the algorithms proposed in the above papers.  Our motivation stems primarily from the connection we establish between delay optimal network control and the throughput optimal BP algorithm.  Our proposed class of enhanced BP algorithms improve delay performance by incorporating a general QSI-dependent bias function into the backpressure calculation, thereby mitigating the myopia of the traditional BP algorithm.

\section{Network Model}\label{sec:system_model}

In this section, we establish the network model. 
Consider a (wireline or wireless) multi-hop network modeled by a directed graph $\mathcal
G=(\mathcal N, \mathcal L)$, where $\mathcal N$ and $\mathcal L$
denote sets of $N$ nodes and  $L$ directed links, respectively.
Time is slotted 
with slots indexed by $t\in \{0,1,2,\cdots\}$. The slot duration  is $\Delta> 0$ (sec). 
Data entering the network is associated with a particular
commodity. 
Let $\mathcal C$ represent the set of $C$ commodities in
the network.  Assume that there is one destination node $dest(c )$ for each commodity $c \in {\cal C}$.
Let $A^{(c)}_n(t)\Delta\geq 0$ denote the amount of exogenous
arrivals (bits) for commodity $c$ to node $n$ during slot $t$, assumed to enter the transmission buffer at 
the end of slot $t$.  Assume that 
 $A^{(c)}_n(t)\in[0, A^{(c)}_{n,\max}]$ is i.i.d. with respect to (w.r.t.) $t$ with arrival rate
$\lambda^{(c)}_n\triangleq \mathbb E[A^{(c)}_n(t)] $ (bit/sec), where $ A^{(c)}_{n,\max}<\infty$. 
 In addition, assume that processes  $\{A^{(c)}_n(t)\}$ for different node-commodity pairs are mutually independent. Let $\mathbf A(t) \triangleq (A^{(c)}_n(t))$ and $\boldsymbol \lambda  \triangleq (\lambda^{(c)}_n)$.

Let $S(t)\in \mathcal S$ denote the topology state of the
network in slot $t$, where $\mathcal S$ is the finite topology state space.
The topology state $S(t)$ can be used to model channel fading in wireless networks. 
Assume $S(t)$ is i.i.d. w.r.t. $t$. Let $I(t)\in \mathcal I$ 
denote the resource allocation action at slot $t$, where $\mathcal I$ is  
the bounded resource allocation action space.
The resource allocation action $I(t)$ may reflect  a set of power allocations or a set of conflict constraints in wireless networks.  
Let $R_{ab}\left(S(t),
I(t)\right)\geq0 $ denote the transmission rate (bit/sec) over link $(a,b)$ under   $S(t)$ and  $I(t)$, where $R_{ab}\left(S(t) , I(t)\right) =0$ if $(a,b)\not \in \mathcal L$. 
Assume $
R_{ab}\left(S(t) , I(t)\right) \leq R_{\max}$ for all $(a,b)\in \mathcal L$, $S(t)\in \mathcal S$ and  $I (t)\in \mathcal I$, where $R_{\max}<\infty$ is
an upper bound on the maximum transmission rate over any link. 
 Let $\nu^{(c)}_{ab}(t)\Delta\geq 0$ 
represent the amount of commodity $c$ data (bits) delivered over link $(a,b)$ during slot $t$, satisfying:
\begin{align}
&\sum_{c\in \mathcal C} \nu^{(c)}_{ab}(t)\leq R_{ab}\big(S(t), I(t)\big),\
\forall (a,b)\in \mathcal L, \ c\in \mathcal C\label{eqn:rout_cost_sum}\\
&\nu^{(c)}_{ab}(t)=0, \ \forall  (a,b)\not\in \mathcal
L^{(c)}, \ c\in \mathcal C\label{eqn:rout_cost_non_neg}
\end{align}
where $\mathcal L^{(c)}$ is the set of $L^{(c)}$ links that are allowed to
transmit commodity $c$ data.  Let $\mathcal R$ denote the bounded routing action space, which is the bounded set of non-negative $\boldsymbol \nu(t)\triangleq (\nu^{(c)}_{ab}(t))$ satisfying \eqref{eqn:rout_cost_sum} and \eqref{eqn:rout_cost_non_neg}, for all  $S(t)\in \mathcal S$ and  $I (t)\in \mathcal I$.

Data corresponding to different commodities are queued separately at each node, in 
buffers of infinite size.  Let $U^{(c)}_n(t)\geq 0$ denote the amount of commodity $c$ data (bits) at
node $n$ at the beginning of slot $t$  in the network layer.  Let $\mathbf
U(t)\triangleq(U^{(c)}_n(t)) \in \mathcal U$ denote  the network layer queue state information (QSI) at the beginning of slot $t$, where $\mathcal U$ denotes the nonnegative QSI state space.   Any data successfully delivered to its destination is assumed to exit the network layer. Thus, for each commodity $c\in \mathcal C$,  we set $U^{(c)}_n(t)=0$ for all $t$, if node $n$ is the destination node of commodity $c$.
For each commodity $c\in \mathcal C$ and node $n\in \mathcal N$, $n\neq dest(c)$, the  network queue dynamics  satisfies:\footnote{Due to the constraint in \eqref{eqn:rout_cost_non_neg}, the summations in \eqref{eqn:queue_dyn-mdp} can be written as summations over all node indices. Note that we assume exogenous arrivals during a slot arrive into the transmission buffer at the end of the slot.}
\begin{align}
&U^{(c)}_n(t+1)\label{eqn:queue_dyn-mdp}\\
=&
U^{(c)}_n(t)-\sum_{b\in \mathcal N}\nu^{(c)}_{nb}(t)\Delta+A^{(c)}_n(t)\Delta
+\sum_{a\in \mathcal N}\nu^{(c)}_{an}(t)\Delta\nonumber
\end{align}
for all $t$. 
Note that $\sum_{b\in \mathcal N}\nu^{(c)}_{nb}(t)\Delta$  bits  are removed from the buffer at node $n$ for commodity $c$ before $A^{(c)}_n(t)\Delta$  and $\sum_{a\in \mathcal N}\nu^{(c)}_{an}(t)\Delta$ bits arrive. Thus, for all $t$, we require: 
\begin{align}
\sum_{b\in \mathcal N}\nu^{(c)}_{nb}(t)\Delta\leq U^{(c)}_n(t), \ \forall n\in \mathcal N, \ c\in \mathcal C.  \label{eqn:bound}
\end{align}

In the following, we introduce some basic definitions.


\begin{Def}  A feasible stationary policy $\omega
: \mathcal S \times  \mathcal U\to \mathcal I\times \mathcal R$ is a mapping from the system (topology and QSI) state space to the system (resource allocation and routing) action space.  Given system state $(s,\mathbf u)\in \mathcal S \times  \mathcal U$, $\omega$ determines the action $( I, \boldsymbol \nu
 )=\omega(s,\mathbf u)\in  \mathcal I\times \mathcal R$, where  $I$ and $\boldsymbol \nu$ satisfy $I \in \mathcal I$ and 
\begin{align}
&\sum_{c\in \mathcal C} \nu^{(c)}_{ab}\leq R_{ab}\big(s, I\big),\
\forall (a,b)\in \mathcal L,\ c\in \mathcal C\label{eqn:rout_cost_sum-bellman}\\
&\nu^{(c)}_{ab}=0, \ \forall  (a,b)\not\in \mathcal
L^{(c)}, \ c\in \mathcal C\label{eqn:rout_cost_non_neg-bellman}\\
&\sum_{b\in \mathcal N}\nu^{(c)}_{nb}\Delta\leq u^{(c)}_n,\ \forall n\in \mathcal N, \ c\in \mathcal C.\label{eqn:bound-bellman}
\end{align}\label{Def:policy}
\end{Def}

Next, we define the notions of network stability and the network stability region. 

\begin{Def}[Network Layer Queue Stability] 
A single network layer queue is  stable if $ \limsup_{t\rightarrow +\infty}\frac{1}{t}\sum_{\tau=0}^{t-1}
\mathbb E [U_n^{(c)}(\tau)]$ $<\infty.$ A network of network layer queues is  stable if all individual network layer queues of the network are  stable.
\end{Def}

\begin{Def}[Network Stability Region] The network stability region $\Lambda$   is the closure of the set of all arrival rates $\boldsymbol \lambda$ for which all network layer queues  can be stabilized 
by some  feasible policy.\footnote{Here, the feasible policy is not required to be stationary.}
\end{Def}

\section{Connection between Throughput Optimal BP algorithm and Delay Optimal Control}

In this section, we consider the case where  $\boldsymbol \lambda  \in \text{int}(\Lambda)$ and study the connection between  the throughput optimal BP algorithm and the delay optimal control.

\subsection{Delay Optimal Control Problem}\label{subsec:delay-opt-problem}

Under a given stationary policy $\omega
$, the induced random process $\{(S(t),\mathbf U(t))\}$ is a controlled markov chain with transition probabilities given by:
\begin{align}
&P_{(s,\mathbf u),(s',\mathbf u')}(I, \boldsymbol \nu
)\nonumber\\
\triangleq &\Pr[(S(t+1),\mathbf U(t+1))=(s',\mathbf u')|(S(t),\mathbf U(t))=(s,\mathbf u), I, \boldsymbol \nu
] \nonumber\\
=&\Pr[S(t+1)=s']P_{(s,\mathbf u),\mathbf u'}(I, \boldsymbol \nu
)\label{eqn:tran-prob}
\end{align}
where $P_{(s,\mathbf u),\mathbf u'}(I, \boldsymbol \nu
)\triangleq \Pr[\mathbf U(t+1)=\mathbf u'|(S(t),\mathbf U(t))=(s,\mathbf u), I, \boldsymbol \nu
]$.
Note that $P_{(s,\mathbf u),(s',\mathbf u')}(I, \boldsymbol \nu
)$  denotes the probability  that the next state will be $(s',\mathbf u')\in \mathcal S \times  \mathcal U$ given that the current state is $(s,\mathbf u)\in \mathcal S \times  \mathcal U$  and the control action is $( I, \boldsymbol \nu
)=\omega(s,\mathbf u)$.  In addition, $P_{(s,\mathbf u),\mathbf u'}(I, \boldsymbol \nu
)=\sum_{s'\in \mathcal S} P_{(s,\mathbf u),(s',\mathbf u')}(I, \boldsymbol \nu
)$.

We now formulate the delay optimal control problem.
\begin{Prob} [Delay Optimal Control Problem]
\begin{align}
\min_{\omega
}\limsup_{t\to\infty}\frac{1}{t}\sum_{\tau=0}^{t-1}\sum_{n\in \mathcal N,c\in \mathcal C} \mathbb
E[U^{(c)}_n(\tau)]\label{eqn:delay-opt-prob}
\end{align}
where $(S(0),\mathbf U(0))\in \mathcal S \times  \mathcal U$  and $\omega$ is a feasible stationary policy defined  in Definition~\ref{Def:policy}.  
Note that the expectation is taken w.r.t. the probability measure induced by the control policy $\omega
$. Problem~\ref{Prob:opt} is an infinite horizon average cost problem  \cite{Bertsekas:2007}.  In the remainder of this section, we restrict our attention to stabilizing policies.\label{Prob:opt}
\label{Prob:opt}
\end{Prob}

\subsection{Delay Optimal Policy}\label{subsec:delay-opt-policy}
In Sections~\ref{subsec:delay-opt-policy}, \ref{subsec:asymp-delay-opt} and \ref{subset:connection},  for technical tractability, we assume $A^{(c)}_n(t)$ and $\nu^{(c)}_n(t)$ both take on nonnegative rational values for  all $t$, $\Delta$ is a nonnegative rational number, and $\mathcal I$ and $\mathcal R$ are finite.  These assumptions ensure that 
$  \mathcal U$ is countably infinite. 
Under certain conditions (specified below), a delay optimal policy $\omega
^*$ can be obtained by solving the following Bellman equation.

\begin{Lem} [Bellman Equation]
Assume\footnote{Assumption 4. 6. 2 and Assumption 4.6.3 in \cite{Bertsekas:2007} provide the conditions for the existence of $d$ and $V(\cdot)$.} that a scalar $d$ and a real-valued function $V(\cdot)$ solve the   Bellman equation:
\begin{align}
&d +V(\mathbf u)=\nonumber\\
&\sum_{s\in \mathcal S}\Pr[S=s]\left[\min_{I, \boldsymbol \nu
} \left\{\sum_{\substack{n\in \mathcal N\\c\in \mathcal C}}u^{(c)}_n+\sum_{\mathbf u'\in \mathcal U}P_{(s,\mathbf u),\mathbf u'}(I, \boldsymbol \nu
)V(\mathbf u')
\right\}\right]\label{eqn:Bellman}
\end{align}
for all $\mathbf u\in \mathcal U$ and furthermore $V(\cdot)$  satisfies:  
\begin{align}
\lim_{t\to \infty}\frac{1}{t}\mathbb E[V(\mathbf U(t)|(S(0),\mathbf U(0))=(s,\mathbf u),\omega
]=0\label{eqn:Bellman-bound}
\end{align}
for all $\omega
$ and $(s,\mathbf u)\in \mathcal S \times  \mathcal U$, where $I$ and $\boldsymbol \nu$ in \eqref{eqn:Bellman} satisfy $I \in \mathcal I$, \eqref{eqn:rout_cost_sum-bellman}, \eqref{eqn:rout_cost_non_neg-bellman} and \eqref{eqn:bound-bellman}.  
Then, $d=\min_{\omega
}\limsup_{t\to\infty}\frac{1}{t}\sum_{\tau=0}^{t-1}\sum_{n\in \mathcal N,c\in \mathcal C} \mathbb
E[U^{(c)}_n(\tau)]$ is the optimal value to Problem~\ref{Prob:opt} for all initial $(S(0),\mathbf U(0))\in \mathcal S\times\mathcal U$ and $V(\cdot)$ is called the  value  function (potential function). Furthermore, if\footnote{Note that $V(\mathbf u)$ captures the average delay cost starting from $\mathbf u$.} 
\begin{align}
\omega^*(s,\mathbf u)=\arg\min_{I, \boldsymbol \nu
} \sum_{\mathbf u'\in \mathcal U}P_{(s,\mathbf u),\mathbf u'}(I, \boldsymbol \nu
)V(\mathbf u')\label{eqn:bellman-opt-policy}
\end{align}
for all $(s,\mathbf u)\in \mathcal S \times  \mathcal U$,  where $I$ and $\boldsymbol \nu$ in \eqref{eqn:bellman-opt-policy} satisfy $I \in \mathcal I$, \eqref{eqn:rout_cost_sum-bellman}, \eqref{eqn:rout_cost_non_neg-bellman} and \eqref{eqn:bound-bellman}, then $\omega^*$ is the delay optimal policy achieving the optimal value $d$.\label{Lem:Bellman}
\end{Lem} 
\begin{proof} Please refer to Appendix A for the proof.
\end{proof}

The fact that the optimal value $d$ does not change with $(S(0),\mathbf U(0))$ implies that the optimal policy $\omega^*$ is a unichain policy\cite{Bertsekas:2007}. From Lemma~\ref{Lem:Bellman}, we can see that  $\omega
^*$ given by  \eqref{eqn:bellman-opt-policy} depends on  $\mathbf u$ through the value function $V(\cdot)$. Obtaining $V(\cdot)$ involves solving the Bellman equation in \eqref{eqn:Bellman}  for all $\mathbf u\in \mathcal U$, which does not admit a  closed-form solution in general. Brute force numerical solutions 
such as value iteration and policy iteration 
are not practical for  multi-hop queueing networks \cite{CuiDelayTutorial:2012,Bertsekas:2007} and  do not yield many design insights.

\subsection{Asymptotically Delay Optimal Policy}\label{subsec:asymp-delay-opt}
In Sections~\ref{subsec:asymp-delay-opt} and \ref{subset:connection}, for technical tractability, we suppose that for all  slot durations $\Delta>0$, a scalar $d$ and a real-valued function $V(\cdot)$ exist in  Lemma~\ref{Lem:Bellman}, and $V(\cdot)$   is twice differentiable.  Let $V_{n,(c)}'(\cdot)\triangleq \frac{\partial V}{\partial u_n^{(c)}}(\cdot)$. 

It is difficult to directly study the features of delay optimal control policy for a general multi-hop network by investigating the properties of the value function $V(\cdot)$. We therefore study the asymptotic features of the delay optimal control. Specifically, using Taylor's theorem for $V(\cdot)$\cite{CuitwohopIT2015}, we can show that minimizing the expected value function in the delay optimal policy (cf. \eqref{eqn:bellman-opt-policy}) is the same as maximizing the gradient descent of the value function (cf.~\eqref{eqn:prob-asymp}), as $\Delta\to 0$ (Please see Appendix B for the proof). 
Consider the policy 
 \begin{align}
&\omega^{\dagger}(s,\mathbf u)\nonumber\\
=&\arg\max_{I, \boldsymbol \nu
}\sum_{(a,b)\in \mathcal L}\sum_{c\in \mathcal C}\nu_{ab}^{(c)}\left(V_{a,(c)}'(\mathbf u)- V_{b,(c)}'(\mathbf u)\right)\label{eqn:prob-asymp}
\end{align}
for all $(s,\mathbf u)\in \mathcal S \times  \mathcal U$, where $I$ and $\boldsymbol \nu$ in \eqref{eqn:prob-asymp} satisfy $I \in \mathcal I$, \eqref{eqn:rout_cost_sum-bellman}, \eqref{eqn:rout_cost_non_neg-bellman} and \eqref{eqn:bound-bellman}. Suppose that policy $\omega^{\dagger}$ is a unichain policy. Denote the average queue length under $\omega^{\dagger}$ by $d^{\dagger}=\limsup_{t\to\infty}\frac{1}{t}\sum_{\tau=0}^{t-1}\sum_{n\in \mathcal N,c\in \mathcal C} \mathbb
E[U^{(c)}_n(\tau)]$, where the expectation is taken w.r.t. the probability measure induced by $\omega^{\dagger}$. 
Note that $\omega^{\dagger}$,  $d^{\dagger}$, $\omega^*$ and $d$ are all functions of the slot duration $\Delta$, as $d$ and $V(\cdot)$ in  Lemma~\ref{Lem:Bellman}  depend on $\Delta$ through the transition probabilities. Therefore, we also write $d^{\dagger}$ and $d$ as $d^{\dagger}(\Delta)$ and $d(\Delta)$, respectively, when studying the scaling behavior w.r.t. $\Delta$. 

Next, we show that   $\omega^{\dagger}$ is asymptotically delay optimal.

\begin{Lem} [Asymptotically Delay Optimal Policy] For all slot durations $\Delta>0$, suppose policy $\omega^{\dagger}$ is a unichain policy.   Then, we have $ d^{\dagger}(\Delta)-d(\Delta)=o (\Delta)$  as  $\Delta\to 0$.\label{Lem:asym-opt-policy}
\end{Lem}
\begin{proof} Please refer to Appendix B for the proof.
\end{proof}

Note that by Lemma~\ref{Lem:asym-opt-policy}, we have  $d^{\dagger}(\Delta)- d(\Delta)\to 0$ as  $\Delta\to 0$. This shows that policy $\omega^{\dagger}$ is asymptotically delay optimal when the scheduling slot duration is small. It can be easily shown that policy $\omega
 ^{\dagger}$  chooses resource allocation and routing action according to the following corollary.

\begin{Cor} [Asymptotically Delay Optimal Policy]  Given the observed   $(s,\mathbf u)\in \mathcal S \times  \mathcal U$, let $I^{\dagger}$ and $ \boldsymbol \nu^{\dagger}$  be defined as in \eqref{eqn:res_allo_opt-asym} and \eqref{eqn:rout_allo_opt-asym}.  If $\boldsymbol \nu^{\dagger}$ satisfies \eqref{eqn:bound-bellman}, then $\omega
^{\dagger}$ chooses the resource allocation and routing action $\omega
 ^{\dagger}(s,\mathbf u)=(I^{\dagger}, \boldsymbol \nu^{\dagger})$.

\textbf{Resource Allocation}: 
For each link $(a,b)\in \mathcal
L$ and commodity $c\in \mathcal C$, let
\begin{align}
&\delta V_{ab}^{(c)}(\mathbf u)\triangleq V_{a,(c)}'(\mathbf u)- V_{b,(c)}'(\mathbf u)
\label{eqn:weight-asym}
\end{align}
denote the asymptotically delay optimal backpressure of link $(a,b)$ w.r.t. commodity $c$.  Let 
$c^{\dagger}_{ab}(\mathbf u)\triangleq \arg
\max_{c\in \{c: (a,b)\in \mathcal L^{(c)}\}}\delta V_{ab}^{(c)}(\mathbf u)$ and $\delta V_{ab}^{\dagger}(\mathbf u)
\triangleq \left(\delta V_{ab}^{(c^{\dagger}_{ab}(\mathbf u))}(\mathbf u)\right)^+$, where $(x)^+\triangleq\max\{x, 0\}$. 
Choose the resource allocation
action  as: 
\begin{align}
I^{\dagger}=\arg\max_{I \in \mathcal I}\sum_{(a,b)\in \mathcal L}\delta V_{ab}^{\dagger}(\mathbf u)
R_{ab}\big(S, I\big). \label{eqn:res_allo_opt-asym}
\end{align}

\textbf{Routing}: For each link $(a,b)\in \mathcal
L$ and commodity $c\in \mathcal C$,  choose the routing action according to: 
\begin{align}
\nu^{(c){\dagger}}_{ab}
=&\begin{cases} R_{ab}\big(S, I^{\dagger}\big),
&  \delta V_{ab}^{\dagger}(\mathbf u)>0\  \text{and}\ c=c^{\dagger}_{ab}(\mathbf u) \\
0, &\text{otherwise.}
\end{cases}\label{eqn:rout_allo_opt-asym}
\end{align}
\label{Cor:asym-opt-policy}
\end{Cor}

\subsection{Connection to the BP Algorithm}\label{subset:connection}

We now discuss the connection between the asymptotically delay optimal policy and  the throughput optimal 
BP algorithm. By Corollary~\ref{Cor:asym-opt-policy}, we can see that  the asymptotically delay optimal control, obtained using dynamic programming, shares striking similarities with the 
BP algorithm\cite{Tassiulas-Ephremides:1992-2, Georgiadis-Neely-Tassiulas:2006}, obtained using Lyapunov drift techniques. Specifically, the two algorithms both base resource allocation and routing on a backpressure calculation. On the other hand, the two algorithms differ in the form of the backpressure calculation employed.  In the BP algorithm, the backpressure of  link $(a,b)$ w.r.t. commodity $c$ is derived from the difference
between  the queue lengths of commodity $c$ at the two end nodes $a$ and $b$ of the link, i.e., $u_a^{(c)}-u_b^{(c)}$.  In the asymptotically delay optimal control algorithm, the backpressure of link $(a,b)$ w.r.t.  commodity $c$ is derived from the differences of the derivatives of the value function at the two end nodes $a$ and $b$ of the link, i.e., $V_{a,(c)}'(\mathbf u)- V_{b,(c)}'(\mathbf u)$. 

 The following lemma summarizes the property of the value function. First, we introduce the BP control $\omega^{\ddagger}$ as follows:
 \begin{align}
\omega^{\ddagger}(s,\mathbf u)=&\arg\max_{I, \boldsymbol \nu
}\sum_{(a,b)\in \mathcal L}\sum_{c\in \mathcal C}\nu_{ab}^{(c)}\left(u^{(c)}_a- u^{(c)}_b\right)\label{eqn:prob-DBP}
\end{align}
for all $(s,\mathbf u)\in \mathcal S \times  \mathcal U$, where $I$ and $\boldsymbol \nu$ in \eqref{eqn:prob-DBP} satisfy $I \in \mathcal I$, \eqref{eqn:rout_cost_sum-bellman}, \eqref{eqn:rout_cost_non_neg-bellman} and \eqref{eqn:bound-bellman}. 
We know that the traditional BP algorithm\cite{Tassiulas-Ephremides:1992-2, Georgiadis-Neely-Tassiulas:2006} is closely related to $\omega^{\ddagger}$.\footnote{The BP algorithm does not consider the constraint in  \eqref{eqn:bound-bellman}, as \eqref{eqn:bound-bellman} is automatically satisfied when queue lengths are large and does not matter when dealing with throughput optimality.} 

\begin{Lem} [Comparison between $\omega^{\dagger}$ and $\omega^{\ddagger}$] For some $n\in \mathcal N$ and $c\in \mathcal C$, there does not exist a function $g_n^{(c)}(u_n^{(c)})$, such that $V_{n,(c)}'(\mathbf u)=g_n^{(c)}(u_n^{(c)})$. Furthermore, there exists $(s,\mathbf u)\in \mathcal S \times  \mathcal U$ such that $\omega^{\dagger}(s,\mathbf u)\neq \omega^{\ddagger}(s,\mathbf u)$. \label{Lem:asym-opt-policy-GQSI}
\end{Lem}
\begin{proof} Please refer to Appendix C for the proof.
\end{proof}

Lemma~\ref{Lem:asym-opt-policy-GQSI} shows that   the backpressure term for the asymptotically delay optimal control is  in general a function of the {\em global} QSI.  On the other hand, the BP backpressure term reflects {\em local} QSI.  Therefore, the asymptotically delay optimal control and  the traditional BP algorithm are significantly different.

Suppose that policy $\omega^{\ddagger}$ is a unichain policy. Denote the average queue length under  $\omega^{\ddagger}$ by $d^{\ddagger}=\limsup_{t\to\infty}\frac{1}{t}\sum_{\tau=0}^{t-1}\sum_{n\in \mathcal N,c\in \mathcal C} \mathbb
E[U^{(c)}_n(\tau)]$, where the expectation is taken w.r.t. the probability measure induced by $\omega^{\ddagger}$. We also write $d^{\ddagger}$ as $d^{\ddagger}(\Delta)$.

\begin{Lem} [Comparison between $d^{\ddagger}$ and $d^{\dagger}$] For all slot durations $\Delta>0$, suppose  policy $\omega^{\dagger}$ and policy $\omega^{\ddagger}$  are unichain policies.   Then,    there exists $\epsilon>0$ such that   $d^{\ddagger}(\Delta)-d^{\dagger}(\Delta)\geq \epsilon+o(\Delta)$ as $\Delta \to 0$. \label{Lem:asym-DBP}
\end{Lem}
\begin{proof} Please refer to Appendix D for the proof.
\end{proof}

Lemma~\ref{Lem:asym-DBP} shows that the BP algorithm has worse delay performance than the asymptotically delay optimal control $\omega^{\dagger}$. By Lemma~\ref{Lem:asym-opt-policy} and Lemma~\ref{Lem:asym-DBP}, we have   $d^{\ddagger}(\Delta)-d(\Delta)\geq\epsilon+ o(\Delta)$. Thus, $d^{\ddagger}(\Delta)-d(\Delta)\not \to 0$ as $\Delta\to 0$. Therefore,    the BP algorithm is not  asymptotically delay optimal. 

Lemma~\ref{Lem:asym-opt-policy-GQSI} and Lemma~\ref{Lem:asym-DBP}  suggest a  possible reason for the generally unsatisfactory delay performance of the BP algorithm.  Namely, the BP delay performance suffers from the {\em myopic} nature of the control, which relies only on one-hop queue size differences.  To the best of our knowledge, this is the first work which provides an analytical connection between delay optimal network control and the BP algorithm (throughput optimal network control).  The connection provides a theoretical basis for designing enhanced BP algorithms with improved delay performance via the use of QSI beyond one hop.  Motivated by this connection, we shall present a new class of enhanced BP-based algorithms, for the cases where  $\boldsymbol \lambda  \in  \text{int}(\Lambda)$ and $\boldsymbol \lambda \notin \Lambda$ in Sections~\ref{Sec:in} and \ref{Sec:out}, respectively.

\section{Enhanced BP Algorithms with \\Stabilizable Arrival Rates}\label{Sec:in}

In this section,  we consider the case where $\boldsymbol \lambda  \in \text{int}(\Lambda)$,  and  develop a new class of {\em enhanced BP algorithms}.

\subsection{Network Layer Queue Dynamics}

When $\boldsymbol \lambda  \in \text{int}(\Lambda)$, the arrival data can be directly admitted
to the network. Let $\mu^{(c)}_{ab}(t)\Delta\geq 0$ 
represent the amount of commodity $c$ data (bits) which can be transmitted over link $(a,b)$ during slot $t$. Let $\boldsymbol \mu(t)\triangleq(\mu^{(c)}_{ab}(t))$. Similar to $\boldsymbol \nu(t)$, $\boldsymbol \mu(t)$ also satisfies \eqref{eqn:rout_cost_sum}  and \eqref{eqn:rout_cost_non_neg} (in terms of $\mu^{(c)}_{ab}(t)$ instead of $\nu^{(c)}_{ab}(t)$). Unlike $\boldsymbol \nu(t)$, $\boldsymbol \mu(t)$ does not have to satisfy \eqref{eqn:bound} (in terms of $\mu^{(c)}_{ab}(t)$ instead of $\nu^{(c)}_{ab}(t)$). In other words,  for each node $n\in \mathcal N$ and commodity $c\in \mathcal C$, $\mu^{(c)}_{nb}(t)=\nu^{(c)}_{nb}(t)$ for all $b\in \mathcal N$ when there are enough bits to be removed from  the buffer at node $n$ for commodity $c$, i.e., \eqref{eqn:bound} is satisfied (in terms of $\mu^{(c)}_{ab}(t)$ instead of $\nu^{(c)}_{ab}(t)$). 
Thus,  for each commodity $c\in \mathcal C$ and node $n\in \mathcal N$, $n\neq dest(c)$,  the network queue dynamics satisfies: 
\begin{align}
&U^{(c)}_n(t+1)\label{eqn:queue_dyn}\\
\leq &
\left(U^{(c)}_n(t)-\sum_{b\in \mathcal N}\mu^{(c)}_{nb}(t)\Delta\right)^+ +A^{(c)}_n(t)\Delta
+\sum_{a\in \mathcal N}\mu^{(c)}_{an}(t)\Delta.\nonumber
\end{align}
Inequality holds in \eqref{eqn:queue_dyn} because the actual amount of commodity $c$ data arriving to node $n$ during slot $t$ may be less than $\sum_{a\in \mathcal N}\mu^{(c)}_{an}(t)\Delta$ if the neighboring nodes have little or no commodity $c$ data to transmit. To facilitate the design of throughput optimal control, we consider routing control in terms of $\boldsymbol \mu(t)$ instead of $\boldsymbol \nu(t)$, as in \cite{Tassiulas-Ephremides:1992-2,Georgiadis-Neely-Tassiulas:2006}.

\subsection{Bias Function}

Motivated by  the  connection between the asymptotically delay optimal control and the throughput optimal BP algorithm in Section~\ref{subset:connection},  we now propose a general QSI-dependent bias function to incorporate QSI beyond one hop in order to mitigate the myopic nature of the BP algorithm.

We first present a general nonnegative QSI-dependent bias function for each node $n\in \mathcal N$ and commodity $c\in \mathcal C$:
\begin{align}
f_n^{(c)}(\mathbf
u)=\sum_{k\in \mathcal N}\eta^{(c)}_{nk}(\mathbf
u) \frac{u^{(c)}_{k}}{z_{k}^{(c)}}. \label{eqn:bias-func}
\end{align}
Here, $\eta^{(c)}_{nk}(\mathbf
u)\in [0,1]$ is the weight associated with QSI $u^{(c)}_{k}$ at node $n$, representing the relative importance of $u^{(c)}_{k}$ in the bias at node $n$.   Note that in general, $\eta^{(c)}_{nk}(\mathbf
u)$ is allowed to depend on the global QSI ${\mathbf u}$.  The parameter $z_{k}^{(c)}>0$ is designed to
guarantee network stability.  The proper choice of $z_{k}^{(c)}>0$ will be discussed below in Theorem~\ref{Thm:thpt-opt}.
We can treat $\frac{u^{(c)}_{k}}{z_{k}^{(c)} }$  as a normalized version of $u^{(c)}_{k}$. 
Later, we shall see that  the quantity $u^{(c)}_n+f_n^{(c)}(\mathbf u)$ can be regarded as  a tractable approximation  of $V_{n,(c)}'(\mathbf u)$  (cf. \eqref{eqn:weight} and \eqref{eqn:weight-asym}) in the asymptotically delay optimal policy in \eqref{eqn:prob-asymp}.

While the bias function $f_n^{(c)}(\mathbf
u)$ in \eqref{eqn:bias-func}  is generally written as a function of the global  QSI,
one can choose the bias function to depend only on the local QSI within one hop as follows: 
\begin{align}
f_n^{(c)}\left(\mathbf
u_n^{(c)}\right)=\sum_{k\in \{k:(n,k)\in \mathcal L^{(c)}\}} \eta^{(c)}_{nk}\left(\mathbf
u_n^{(c)}\right) \frac{u^{(c)}_{k}}{z_{k}^{(c)}}\label{eqn:spe-f}
\end{align}
where $\mathbf u_n^{(c)}\triangleq (u_{k}^{(c)})_{(n,k)\in \mathcal L^{(c)}}$ is the local QSI within one hop.

Each specific choice of a bias function $\mathbf f\triangleq (f_n^{(c)})$ corresponds to one enhanced BP algorithm (described in the next subsection), and the amount of QSI contributing to the bias function determines the implementation complexity of the corresponding enhanced BP algorithm.  The form of the bias function  (cf.   \eqref{eqn:bias-func} and \eqref{eqn:spe-f}) is carefully chosen to enable (generalized) throughput optimality (Theorem~\ref{Thm:thpt-opt}) of the enhanced BP algorithms, while at the same time offering a high degree of flexibility in choosing specific enhanced BP algorithms with manageable complexity, distributed implementation, and significantly better delay performance. 
In Section~\ref{sec:ex}, we shall describe two enhanced BP algorithms resulting from two specific choices for  $\mathbf f$   in \eqref{eqn:spe-f} and   \eqref{eqn:bias-func}, respectively, which embody the desirable properties described above.

\subsection{Enhanced BP Algorithm}\label{Subsec:DBP}

We now present a new class of enhanced BP algorithms by incorporating 
the general QSI-dependent  bias functions $\mathbf f$ defined in \eqref{eqn:bias-func} into the BP backpressure calculation.

\begin{Alg}[Enhanced BP] 

Let a set of bias functions $\mathbf f\succeq 0$ be given.\footnote{The notation $\succeq$, $\preceq$, $\succ$, $\prec$ indicate component-wise $\geq$, $\leq$, $>$, $<$.}  At each slot $t$, the network 
controller observes the network layer QSI $\mathbf U(t)$ and the
topology state $S(t)$, and performs the following resource
allocation and routing actions.

\textbf{Resource Allocation}: 
For each link $(a,b)\in \mathcal
L$ and commodity $c\in \mathcal C$, let\footnote{A QSI-independent bias, such as the shortest-path bias used in \cite{Neely-Modiano-Rohrs:2005},  
can easily be incorporated into bias functions.  It can be verified that, with any extra QSI-independent  bias, Theorem~\ref{Thm:thpt-opt}  (Theorem~\ref{Thm:flow-control}) still holds  with a constant shift of $\bar B$ in \eqref{eqn:bound-B} ($\hat B$ in \eqref{eqn:bound-B-flow})\cite{Neely-Modiano-Rohrs:2005}.} 
\begin{align}
&W_{ab}^{(c)}(\mathbf U(t))\nonumber\\
\triangleq &\left(U^{(c)}_a(t)+f_a^{(c)}(\mathbf U(t))\right)-\left(U^{(c)}_b(t)+
f_b^{(c)}(\mathbf U(t))\right)\label{eqn:weight}
\end{align}
denote the enhanced BP backpressure of link $(a,b)$ w.r.t. commodity $c$. Let $c^*_{ab}(\mathbf U(t))\triangleq \arg
\max_{c\in\{c: (a,b)\in \mathcal L^{(c)}\}}W_{ab}^{(c)}(\mathbf U(t))$ and $
W^*_{ab}(\mathbf U(t))
\triangleq \left(W_{ab}^{(c^*_{ab}(\mathbf U(t)))}(\mathbf U(t))\right)^+$. 
Choose the resource allocation
action as follows:
\begin{align}
I(t)=\arg\max_{I \in \mathcal I}\sum_{(a,b)\in \mathcal L}W^*_{ab}(\mathbf U(t))
R_{ab}\big(S(t), I\big). \label{eqn:res_allo_opt}
\end{align}

\textbf{Routing}: For each link $(a,b)\in \mathcal
L$ and commodity $c\in \mathcal C$,  choose: 
\begin{align}
\mu^{(c)}_{ab}(t)
=&\begin{cases} R_{ab}\big(S(t), I(t)\big),
&  W^*_{ab}(t)>0\  \text{and}\ c=c^*_{ab}(t)\\
0, &\text{otherwise}.
\end{cases}\nonumber
\end{align}
\label{Alg:enhanced-DBP}
\end{Alg}

Note that based on $\boldsymbol \mu(t)$, we can choose routing actions $\boldsymbol \nu(t)$. For each node $a\in \mathcal N$ and commodity $c\in \mathcal C$, if $\sum_{b\in \mathcal N_{out,a}^{(c)}}\mu^{(c)}_{ab}(t)\Delta\leq U^{(c)}_a(t)$, we
choose routing action $\nu^{(c)}_{ab}(t)=\mu^{(c)}_{ab}(t)$ for all $b\in \mathcal N_{out,a}^{(c)}$, where $\mathcal N_{out,a}^{(c)}\triangleq \{b:(a,b)\in \mathcal L^{(c)}\}$. When $\sum_{b\in \mathcal N_{out,a}^{(c)}}\mu^{(c)}_{ab}(t)\Delta> U^{(c)}_a(t)$, there are multiple choices for  $\boldsymbol \nu(t)$, which can guarantee generalized throughput optimality shown in Theorem~\ref{Thm:thpt-opt}. 

\subsection{Performance Analysis} \label{Subsec:DBP-performance}

As mentioned above, the traditional BP algorithm can have poor delay performance in networks with light to moderate traffic  loads.  Note that in this situation, there exists a  significant margin between the arrival rate vector and the boundary of the network stability region.  
Algorithm~\ref{Alg:enhanced-DBP}, with    the bias functions chosen according to~\eqref{eqn:bias-func},  can exploit this margin  (or a lower bound on the margin) to incorporate QSI beyond one hop, thereby reducing the average delay while maintaining a generalized notion of throughput optimality.  This result is summarized as follows. For notational simplicity, we assume $\Delta=1$.

\begin{Thm} [Generalized Throughput Optimality of Alg. 1] Given
 $\boldsymbol \epsilon\triangleq(\epsilon_n^{(c)})\succeq  \mathbf 0$  such
that $\boldsymbol \lambda+\boldsymbol \epsilon \in \text{int}(\Lambda)$, 
there exist $\boldsymbol \delta\triangleq (\delta_n^{(c)})\succ  \mathbf 0$ such that
$\boldsymbol \lambda+\boldsymbol \epsilon + \boldsymbol \delta \in 
\Lambda$, and $\mathbf z\triangleq(z_n^{(c)})\succ \mathbf 0$  such that $\boldsymbol \epsilon_{\mathbf z}\triangleq\left(\frac{2R_{\max}L^{(c)}}{z_n^{(c)}}\right)\preceq \boldsymbol \epsilon$.  Then, the queueing network under
Algorithm~\ref{Alg:enhanced-DBP}, with the bias functions chosen according to~\eqref{eqn:bias-func}, satisfies: 
\begin{align}
\limsup_{t\to\infty}\frac{1}{t}\sum_{\tau=0}^{t-1}\sum_{n\in \mathcal N,c\in \mathcal C} \mathbb
E[U^{(c)}_n(\tau)]\leq \frac{N\bar B}{\bar
\beta_{\mathbf z}}\label{eqn:enhanced-DBP-UB}
\end{align}
where  
\begin{align}
\bar B\triangleq& \frac{1}{2N}\sum_{n\in \mathcal N}\left((\mu^{out}_{n,
\max})^2+(A_{n,\max}+\mu^{in}_{n,\max})^2\right)\label{eqn:bound-B}\\
\bar \beta_{\mathbf z}\triangleq &\sup_{\substack{\{(\boldsymbol \epsilon,\boldsymbol \delta):  \boldsymbol \epsilon\succeq \boldsymbol \epsilon_{\mathbf z}, \boldsymbol \delta \succ \mathbf 0,\\ \quad \quad \ \boldsymbol \lambda+\boldsymbol \epsilon+\boldsymbol \delta\in  \Lambda\}}}\min_{ n\in \mathcal N,c\in \mathcal C}\left\{\epsilon_n^{(c)} +\delta_n^{(c)}
-\frac{2R_{\max}L^{(c)}}{z_n^{(c)}}\right\} \label{eqn:bound-e}
\end{align}
with $\mu^{out}_{n, \max}\triangleq \sup_{s\in \mathcal S, I\in
\mathcal I}\sum_{b\in \mathcal N} R_{nb}(s,I)$, $\mu^{in}_{n, \max}\triangleq
\sup_{s\in \mathcal S, I\in \mathcal I}\sum_{a\in \mathcal N} R_{an}(s,I)$, and $
A_{n,\max} \triangleq \sum_{c\in \mathcal C} 
A^{(c)}_{n,\max}$.\label{Thm:thpt-opt}
\end{Thm}
\vspace{0.1in}
\begin{proof} Please refer to Appendix E. 
\end{proof}

\begin{Rem} Theorem~\ref{Thm:thpt-opt} should be interpreted as follows.  When it is given that the arrival rate
vector $\boldsymbol \lambda$ is bounded away from the boundary of the stability region $\Lambda$ by at least $\boldsymbol \epsilon \succ \mathbf 0$, i.e., $\boldsymbol \lambda+\boldsymbol \epsilon \in \text{int}(\Lambda)$, 
one can choose a {\em finite} $\mathbf z \succ \mathbf 0$ such that $\boldsymbol \epsilon_{\mathbf z}\prec \boldsymbol \epsilon$.  
In this case, the limiting
average total queue size under Algorithm~\ref{Alg:enhanced-DBP} is    upper bounded as in~\eqref{eqn:enhanced-DBP-UB}.  Thus, Algorithm~\ref{Alg:enhanced-DBP} stabilizes the network for any arrival rate which is bounded away from the boundary of the stability region by at least $\boldsymbol \epsilon$, for any given $\boldsymbol \epsilon \succ \mathbf 0$. 
When it is only known that  $\boldsymbol \lambda \in \text{int}(\Lambda)$ and no extra margin is given ($\boldsymbol \epsilon = \mathbf 0$), then by Theorem~\ref{Thm:thpt-opt}, $z_n^{(c)}$  must be  chosen to be infinity for all $n\in \mathcal N$ and $c\in \mathcal C$ (i.e., $\mathbf f = \mathbf 0$).  In this case, Algorithm~\ref{Alg:enhanced-DBP} reduces to the traditional BP algorithm, and Theorem~\ref{Thm:thpt-opt} reduces to the traditional throughput optimality result (Lemma 4.1 in~\cite{Georgiadis-Neely-Tassiulas:2006}).
\label{Rem:DBP}
\end{Rem}

The margin (or a lower bound on the margin)  $\boldsymbol \epsilon$ in  Theorem~\ref{Thm:thpt-opt}  may be obtained in several possible ways.  First, traffic measurement is usually performed for practical networks (e.g., at different times of a day).  The difference between the measured peak traffic load during the busy traffic hour (assuming the network remains stable) and the load during an off-peak non-busy traffic hour can serve as a lower bound on the margin during the non-busy traffic hour.  Second, in the case where the arrival rate vector, or an upper bound on the arrival rate vector, can be estimated, one can calculate a lower bound on the margin by solving a linear program (in a distributed manner) maximizing $\min_{n\in \mathcal N,c\in \mathcal C}\epsilon_n^{(c)}$ subject to  $\boldsymbol \lambda+\boldsymbol \epsilon \in \text{int}(\Lambda)$, where $\Lambda$ is characterized in \cite[pp. 32]{Georgiadis-Neely-Tassiulas:2006}.  

Note that the choice of $\mathbf z$ based on $\boldsymbol \epsilon$ as given by Theorem~\ref{Thm:thpt-opt} for throughput optimality  may not be optimal from the viewpoint of improving delay performance.  In Section~\ref{sec:simu}, we shall show using numerical simulations that the proposed enhanced BP algorithms with appropriately chosen parameters  $\mathbf z$ achieve significant delay improvement over the existing BP-based algorithms under small or moderate loading. 

\section{Two Bias Functions}\label{sec:ex}

In this section, we describe two enhanced BP algorithms  resulting from two specific choices for $\mathbf f$.  
These enhanced BP algorithms are designed to ameliorate the myopia of the traditional BP algorithm, thereby improving
its delay performance.  In addition, the enhanced BP algorithms can be implemented in a distributed manner with manageable
complexity.

\begin{table}[h]
\begin{center}
\begin{tabular}{|c|c|c|}
\hline
Algorithm & Computational Complexity & Signaling Overhead \\
  \hline
  BP  (BPbias) &  $\mathcal O(N^2C)$ &  $\mathcal O(N^2C)$\\
  \hline
  BPnxt (BPnxtbias) &  $\mathcal O(N^2C)$ &  $\mathcal O(N^2C)$\\
  \hline
  BPmin  (BPminbias) & $\mathcal O(N^3C)$ &   $\mathcal O(N^3C)$\\
  \hline
  BPSP &  $\mathcal O(N^4C)$ &   $\mathcal O(N^3C)$\\
  \hline
\end{tabular}
  \caption{\small{Comparisons on algorithm complexity.}}
  \label{Tab:comparison}
\end{center}
\end{table}

\subsection{Minimum Next-hop Queue Length Bias (BPnxt)}\label{subsec:DBPnxt}

We consider a local QSI-dependent bias function which is an example of \eqref{eqn:spe-f} and allows the resulting enhanced BP algorithm to incorporate QSI one more hop beyond what is accounted for in the traditional BP algorithm.  Specifically, consider the {\em minimum next-hop queue length bias function} defined as follows. For each node $n\in \mathcal N$ and commodity $c\in \mathcal C$, let $H_{n}^{* (c )}\left(\mathbf u_n^{(c)}\right) \triangleq \min_{k\in \left\{k: (n,k)\in \mathcal L^{(c)}\right\}} u_{k}^{(c)}$ be the minimum next-hop queue length, and 
choose $\eta_{nk}^{(c)} \left(\mathbf u_n^{(c)}\right)$ as follows:
\begin{align}
\eta_{nk}^{(c)} \left(\mathbf u_n^{(c)}\right)=
\begin{cases}
\frac{1}{ N_{\text{nxt},n}^{*(c)}}, & k\in \mathcal N_{\text{nxt},n}^{*(c)}\\
0, & \text{otherwise}
\end{cases} \label{eqn:ex-onehop-eta}
\end{align}
where $ \mathcal N_{\text{nxt},n}^{*(c)}\triangleq \left\{k:u_{k}^{(c)}=H_n^{*(c)}\left(\mathbf u_n^{(c)}\right), (n,k)\in \mathcal L^{(c)} \right\}$ and $  N_{\text{nxt},n}^{*(c)}\triangleq |\mathcal N_{\text{nxt},n}^{*(c)}|$. 
For any given margin  $\boldsymbol \epsilon=(\epsilon_n^{(c)})\succ\mathbf 0$, we choose $\mathbf z=(z_n^{(c)})$, $z_n^{(c)}=  z$ for all $n\in \mathcal N$ and $c\in \mathcal C$, where
\begin{align}
   z\geq\frac{2R_{\max}d_{in}}{\min_{n\in\mathcal N, c\in \mathcal C}\epsilon_n^{(c)}}.\label{eqn:cond-z-ex}
 \end{align} 
Here, $d_{in}$ denotes the largest node in-degree  among all nodes in the graph. Thus, the minimum next-hop queue length bias function is given by: 
\begin{align}
f_n^{(c)}\left(\mathbf u_n^{(c)}\right)=\frac{1}{ z}H_{n}^{* (c )}\left(\mathbf u_n^{(c)}\right).\label{eqn:ex-onehop-f}
\end{align}

Given the bias function in~\eqref{eqn:ex-onehop-f}, and using Algorithm~\ref{Alg:enhanced-DBP}, we now obtain an enhanced BP algorithm, which will be referred to as BPnxt.   We show in Appendix F that for all $z$ satisfying \eqref{eqn:cond-z-ex}, BPnxt stabilizes the network for any $\boldsymbol \lambda$ satisfying $\boldsymbol \lambda+\boldsymbol \epsilon \in \text{int}(\Lambda)$. 


We now briefly discuss the implementation complexity of the BPnxt algorithm,  as illustrated in Table~\ref{Tab:comparison}.    
Since the difference between the traditional BP algorithm and the BPnxt algorithm lies in the backpressure calculation, we focus only on the complexity for implementing~\eqref{eqn:weight}.   Consider the computational complexity first.  For each commodity, each node needs to compute the minimum next-hop queue length bias, which involves a minimization over no more than 
$N$ queue lengths, and the summation of its queue length and the minimum next-hop queue length bias. 
In addition, for each commodity, each node needs to carry out one  subtraction to compute the enhanced BP backpressure of each outgoing link, involving in total no more than $N$ operations for no more than $N$ outgoing links.  Thus, the total computational complexity for the backpressure calculation of BPnxt (over $N$ nodes and $C$ commodities) is $\mathcal O(N^2C)$.  It is easy to see that this is the same in order as the computational complexity for calculating the backpressure in the traditional BP algorithm and the BPbias algorithm in~\cite{Neely-Modiano-Rohrs:2005}.  Next, consider the signaling overhead for the backpressure calculation of BPnxt.  During the signaling phase of each slot, for each commodity, each node needs to report its queue length and the sum of its queue length and its minimum next-hop queue length bias (which can be obtained from the information reported to the node from its next-hop neighbors) to no more than $N$ previous-hop neighbors. Thus, the total signaling overhead is  $\mathcal O(N^2C)$.   Again, it is easy to see that this is the same in order as the signaling overhead for calculating the backpressure in the traditional BP algorithm and  the BPbias algorithm.  
In summary, the BPnxt algorithm has the same order of implementation complexity  as the traditional BP algorithm and the BPbias algorithm.  

\subsection{Minimum Downstream Sum Queue Length Bias (BPmin)}

Next, we consider a global QSI-dependent bias function which is an example of \eqref{eqn:bias-func} and  incorporates multi-hop QSI downstream 
toward the destinations of the respective traffic commodities.  
Consider the {\em minimum downstream sum queue length bias function} defined as follows.  For each node $n\in \mathcal N$ and commodity $c\in \mathcal C$, let $M_{n}^{(c)}$ denote the number of paths  from node $n$ to the destination node of commodity $c$, $dest( c )$.   Let $\mathcal P_{n,m}^{(c)}$ denote the set of nodes on the $m$-th path from node $n$ to node $dest(c)$, where $m=1,\cdots, M_{n}^{(c)}$. The sum queue length excluding node $n$ on path $\mathcal P_{n,m}^{(c)}$  is given by $T_{n,m}^{(c)}\left(\mathbf u\right)\triangleq\sum_{i\in \mathcal P_{n,m}^{(c)}, i\neq n}u_i^{(c)}$.  Let $T_{n}^{* (c )} \left(\mathbf u\right)\triangleq \min_{m=1,\cdots, M_{n}^{( c)}} T_{n,m}^{(c )}\left(\mathbf u\right)$ be the sum queue length {\em excluding node $n$} on the shortest path (in terms of the sum queue length) to $dest( c)$.  For each node $n\in \mathcal N$ and commodity $c\in \mathcal C$, choose $\eta_{nk}^{(c)}(\mathbf u)$ as follows:
\begin{equation}
\eta_{nk}^{(c)}(\mathbf u) = 
\begin{cases}
\frac{1}{ N_{\min,n}^{*(c)}}, & k\in\mathcal P^*, \mathcal P^*\in \mathcal N_{\min,n}^{*(c)} \\
0, & \text{otherwise}
\end{cases}
\label{eqn:ex-sp-eta}
\end{equation}
where 
$$ \mathcal N_{\min,n}^{*(c)}\triangleq \left\{\mathcal P_{n,m}^{(c)}:T_{n,m}^{(c)}\left(\mathbf u\right)= T_{n}^{*(c )}\left(\mathbf u\right), m \in \left\{1,\cdots, M_{n}^{( c)} \right\} \right\}$$ is the set of shortest paths (in terms of the sum queue length) from $n$ to $dest( c)$ and $  N_{\min,n}^{*(c)}\triangleq  | \mathcal N_{\min,n}^{*(c)}| $. As in BPnxt, for any given margin  $\boldsymbol \epsilon=(\epsilon_n^{(c)})\succ \mathbf 0$, we choose $\mathbf z=(z_n^{(c)})$, $z_n^{( c)}=z$ for all $n\in \mathcal N$ and $c\in \mathcal C$, where $z$ satisfies \eqref{eqn:cond-z-ex}.   
Thus, the minimum downstream sum queue length bias function is given by: 
\begin{align}
f_n^{(c)}\left(\mathbf u\right)=\frac{1}{z} T_{n}^{*(c )}\left(\mathbf u\right).
\label{eqn:ex-sp-f}
\end{align}

Given the bias function in~\eqref{eqn:ex-sp-f}, and using Algorithm~\ref{Alg:enhanced-DBP}, we now obtain an enhanced BP algorithm, which will be referred to as BPmin.   We show in Appendix F that for all $z$ satisfying \eqref{eqn:cond-z-ex}, BPmin stabilizes the network for any $\boldsymbol \lambda$ satisfying $\boldsymbol \lambda+\boldsymbol \epsilon \in \text{int}(\Lambda)$. 


We now discuss the implementation complexity of the BPmin algorithm, as illustrated in Table~\ref{Tab:comparison}.   
Again, we focus only on the complexity for implementing the backpressure calculation given in~\eqref{eqn:weight}.  Consider the computational complexity first.  The downstream sum queue length minimization is a shortest path problem where the per-link cost is the instantaneous queue length at the receive node of the link.  Thus, the minimization can be solved in a distributed and parallel manner using the iterative Bellman-Ford algorithm \cite{Datanetwork:92}. 
  In each iteration, for each commodity, each node needs to compute no more than $N$ summations and one minimization over no more than $N$ alternatives.  The number of iterations is no more than $N$.  Thus, for each node and each commodity, the computational complexity of calculating the minimum downstream sum queue length bias is $\mathcal O(N^2)$. 
In addition, for each commodity, each node needs to carry out the summation of its queue length and the minimum downstream sum queue length bias, and then carry out one  subtraction to compute the enhanced BP backpressure of each outgoing link, involving in total no more than $N$ operations for no more than $N$ outgoing links.  Thus, the total computational complexity for the backpressure calculation of BPmin (over $N$ nodes and $C$ commodities) is $\mathcal O(N^3C)$, which is lower than that of the BPSP algorithm in~\cite{YingShakkottai08oncombSPDBPJNET} ($\mathcal O(N^4C)$), but higher than that of the traditional BP algorithm and the BPbias algorithm ($\mathcal O(N^2C)$).  Next, consider the signaling overhead.  For each commodity and each iteration of the Bellman-Ford algorithm, each node needs to report one intermediate shortest path length to no more than $N$ previous-hop neighbors.  When the Bellman-Ford algorithm converges within no more than $N$ iterations,  each node needs to report the sum of its queue length and  the minimum downstream sum queue length bias to no more than $N$ previous-hop neighbors, for each commodity.  
Thus, the total signaling overhead is $\mathcal O(N^3C)$, which is the same in order as that of the BPSP algorithm, but higher in order than that of the traditional BP algorithm and the BPbias algorithm ($\mathcal O(N^2C)$).  
In summary,    the BPmin algorithm has a higher order of implementation complexity than the traditional BP algorithm, the BPbias algorithm and the BPnxt algorithm, but a lower order of implementation complexity than the BPSP algorithm.  

\section{Enhanced Joint Flow Control and BP Algorithms with Arbitrary Arrival Rates}\label{Sec:out}

In this section, we consider the case where  $\boldsymbol \lambda  \notin\Lambda$, and develop a new class of  {\em enhanced joint flow control and BP algorithms}.

\subsection{Transport Layer and Network Layer Queue Dynamics}

When $\boldsymbol \lambda  \notin \Lambda$, the network cannot be stabilized by any feasible resource allocation and routing policy.  Rather, in order to stabilize the network, a flow controller must be placed in front of each network layer queue at the source nodes to control the amount of  data admitted into the network layer. Newly arriving data first enter  transport layer storage reservoirs before being admitted to the network layer\cite{Georgiadis-Neely-Tassiulas:2006}.   Let $Q_{n,\max}^{(c)}$ and $Q_n^{(c)}(t)$ denote  the transport layer  buffer size and the QSI of commodity $c$ data (bits) at node $n$ at the beginning of slot $t$, respectively.  The buffer size $Q_{n,\max}^{(c)}$ can be infinite or finite (possibly zero).  Similarly, for each commodity $c\in \mathcal C$,  we set $Q_{n,\max}^{(c)}=0$ and $Q^{(c)}_n(t)=0$ for all $t$, if node $n$ is the destination node of commodity $c$. 
Let $r^{(c)}_n(t)\Delta\geq 0$ denote the amount of data  admitted to the network layer queue of commodity $c$ data (bits) at node $n$ from the transport layer queue during slot $t$. Thus, we require $r^{(c)}_n(t)\Delta\leq Q^{(c)}_n(t)$. We assume $r^{(c)}_n(t)\leq r^{(c)}_{n,\max}$, where $r^{(c)}_{n,\max}$ is a   positive constant which limits the burstiness of the admitted arrivals to the network layer\cite{Georgiadis-Neely-Tassiulas:2006}.   For each commodity $c\in \mathcal C$ and node $n\in \mathcal N$, $n\neq dest(c)$,  we have  the following transport  and network layer queue dynamics:
\begin{align}
&Q^{(c)}_n(t+1) \label{eqn:queue_dyn-trans}\\
=& \min \left\{ Q^{(c)}_n(t)-r^{(c)}_n(t)\Delta +A^{(c)}_n(t)\Delta, Q^{(c)}_{n,\max} \right\}
\nonumber\\
&U^{(c)}_n(t+1) \label{eqn:queue_dyn-flow} \\
\leq &\left(U^{(c)}_n(t)-\sum_{b\in \mathcal N}\mu^{(c)}_{nb}(t)\Delta\right)^+ +r^{(c)}_n(t)\Delta+\sum_{a\in \mathcal N}\mu^{(c)}_{an}(t)\Delta.
\nonumber
\end{align}

\subsection{Enhanced Joint Flow Control and BP Algorithm}

The goal of the flow control is to support a portion of the traffic demand $\boldsymbol \lambda$ which maximizes the sum of utilities when $\boldsymbol \lambda  \notin \Lambda$.  Let $h^{(c)}_n(\cdot)$ be the utility function associated with the input commodity $c$ data at node $n$.  Assume  $h^{(c)}_n(\cdot)$ is non-decreasing, concave, continuously differentiable and non-negative. Define  a  $\boldsymbol \theta$-optimal admitted rate as follows:
\begin{align}
\overline {\mathbf r}^*(\boldsymbol \theta)=\arg\max_{\overline{\mathbf r}}\quad  & \sum_{n\in \mathcal N,c\in \mathcal C}h^{(c)}_n\left(\overline r^{(c)}_n\right)\label{eqn:eps-opt-prob}\\
s.t. \quad &  \overline{\mathbf r}+\boldsymbol \theta \in \Lambda\label{eqn:stable}\\
& \mathbf 0 \preceq \overline{\mathbf r} \preceq \boldsymbol \lambda\label{eqn:demand}
\end{align}
where $ \overline{\mathbf r}^*(\boldsymbol \theta)\triangleq (\overline r^{*(c)}_n(\boldsymbol \theta))$, $ \overline{\mathbf r}\triangleq (\overline r^{(c)}_n)$ and $\mathbf 0\preceq\boldsymbol \theta\triangleq (\theta^{(c)}_n)\in \Lambda $.  The constraint in \eqref{eqn:stable} ensures that the admitted rate to the network layer is bounded away from the boundary of the network stability region by $\boldsymbol \theta$. 
Due to the non-decreasing property of the utility functions, the maximum sum utility over all $\boldsymbol \theta$ is achieved at $\overline {\mathbf r}^*(\mathbf 0)$ when  $\boldsymbol \theta=\mathbf 0$.

We now develop a new class of enhanced joint flow control and BP algorithms that yield a throughput vector which can be arbitrarily close to the optimal solution $\overline {\mathbf r}^*(\mathbf 0)$. Following \cite[pp. 90]{Georgiadis-Neely-Tassiulas:2006}, we introduce the auxiliary variables $\gamma_n^{(c)}(t)$ and the virtual queues $Y_n^{(c)}(t)$ for all $n\in \mathcal N$ and $c\in \mathcal C$.\footnote{Note that the flow control part  of Algorithm~\ref{Alg:enhanced-flow-DBP} is the same as that in the traditional joint flow control and BP algorithm in \cite[pp. 90]{Georgiadis-Neely-Tassiulas:2006}. The difference lies in the resource control and routing part. We describe the flow control part here for the purpose of completeness.}

\begin{Alg}[Enhanced Joint Flow Control and BP]  Let a set of bias functions $\mathbf f\succeq 0$ be given.
In each slot $t$, the 
controllers observe the network layer QSI $\mathbf U(t)$, virtual QSI $\mathbf Y(t)$ and
topology state $S(t)$, and performs the following flow control, resource
allocation and routing actions.

\textbf{Flow Control}: For each node $n\in \mathcal N$ and commodity $c\in \mathcal C$, the flow controller observes the transport layer QSI $Q^{(c)}_n(t)$ and the virtual QSI $Y^{(c)}_n(t)$, and chooses the admitted data rate at slot $t$, which also serves as the output rate of the corresponding virtual queue:
\begin{align}
r^{(c)}_n(t)=
\begin{cases}
\min\left\{Q^{(c)}_n(t)/\Delta, r^{(c)}_{n,\max}\right\},  & Y^{(c)}_n(t)>U^{(c)}_n(t)\\
0,& \text{otherwise}.
\end{cases}\nonumber
\end{align}
The flow controller then chooses the auxiliary variable, which serves as the input rate  to the corresponding virtual queue:
\begin{align}
\gamma^{(c)}_n(t)=\arg\max_{\gamma}\quad & M h^{(c)}_n(\gamma)-Y^{(c)}_n(t)\gamma\Delta\label{eqn:solu-gamma}\\
s.t. \quad & 0\leq \gamma\leq r^{(c)}_{n,\max}\nonumber
\end{align}
where $M>0$ is a control parameter which affects the utility-delay tradeoff of the algorithm. 
Based on the chosen $r^{(c)}_n(t)$ and $\gamma^{(c)}_n(t)$, the  transport
layer QSI is updated according to \eqref{eqn:queue_dyn-trans} and the virtual QSI is  updated according to: 
\begin{align}
Y^{(c)}_n(t+1)=& \left(Y^{(c)}_n(t)-r^{(c)}_n(t)\Delta\right)^+ +\gamma^{(c)}_n(t)\Delta\label{eqn:queue_dyn-virtual}
\end{align}
where $Y^{(c)}_n(0)=0$ for all $n\in\mathcal N$ and $c\in \mathcal C$. 

\textbf{Resource Allocation and Routing}:  Same as Algorithm~\ref{Alg:enhanced-DBP}. \label{Alg:enhanced-flow-DBP}
\end{Alg}

\subsection{Performance Analysis}

The following theorem summarizes the utility-delay tradeoff of Algorithm~\ref{Alg:enhanced-flow-DBP}. For notational simplicity, we assume $\Delta=1$.

\begin{Thm}[Utility-Delay Tradeoff of Alg.~\ref{Alg:enhanced-flow-DBP}] For an arbitrary arrival rate vector, for any transport layer buffer size, and for any control parameter $M>0$, given  $\boldsymbol\epsilon \triangleq(\epsilon_n^{(c)}) \in \text{int}(\Lambda)$, there exist   
$\boldsymbol \delta\triangleq (\delta_n^{(c)})\succ  \mathbf 0$ such that
$\boldsymbol \epsilon + \boldsymbol \delta \in 
\Lambda$, and $\mathbf z\triangleq (z_n^{(c)})\succ \mathbf 0$  such that $\boldsymbol \epsilon_{\mathbf z}\triangleq\left(\frac{2R_{\max}L^{(c)}}{z_n^{(c)}}\right)\preceq \boldsymbol \epsilon$.  Then, the queueing network under
Algorithm 2, with the bias  functions chosen according to~\eqref{eqn:bias-func}, satisfies:
\begin{align}
\limsup_{t\to\infty}\frac{1}{t}\sum_{\tau=0}^{t-1}\sum_{n\in \mathcal N,c\in \mathcal C} \mathbb
E[U^{(c)}_n(\tau)]&\leq \frac{N\hat B+MH_{\max}}{
\hat\beta_{\mathbf z}}\label{eqn:enhanced-flow-DBP-U}\\
\liminf_{t\to\infty}\sum_{n\in \mathcal N,c\in \mathcal C} h^{(c)}_n\left(\overline r^{(c)}_n(t)\right)\geq &
\sum_{n\in \mathcal N,c\in \mathcal C}
h^{(c)}_n\left( \overline r^{*(c)}_n\left(\boldsymbol \epsilon_{\mathbf z}\right)\right)\nonumber\\
&-\frac{N\hat B}{M}
\label{eqn:enhanced-flow-DBP-g}
\end{align}
where 
\begin{align}
&\hat B\triangleq \frac{1}{2N}\sum_{n\in \mathcal N}\left((\mu^{out}_{n,
\max})^2+(r_{n,\max}+\mu^{in}_{n,\max})^2+2(r_{n,\max})^2\right)\label{eqn:bound-B-flow}\\
&\hat \beta_{\mathbf z}\triangleq \sup_{\substack{\{(\boldsymbol \epsilon,\boldsymbol \delta):  \boldsymbol \epsilon\succeq \boldsymbol \epsilon_{\mathbf z}, \boldsymbol \delta \succ \mathbf 0, \\  \quad \ \boldsymbol \epsilon+\boldsymbol \delta\in  \Lambda\}}}\min_{ n\in \mathcal N,c\in \mathcal C}\left\{\epsilon_n^{(c)}+\delta_n^{(c)}
-\frac{2R_{\max}L^{(c)}}{z_n^{(c)}}\right\}\label{eqn:bound-e-flow}
\end{align} 
with $r_{n,\max} \triangleq  \sum_{c\in \mathcal C}
r^{(c)}_{n,\max}$,  $H_{\max}\triangleq\sum_{n\in \mathcal N,c\in \mathcal C}
h^{(c)}_n\left(r^{(c)}_{n,\max}\right)$, and  $\overline r^{(c)}_n(t)\triangleq \frac{1}{t}\sum_{\tau=0}^{t-1}\mathbb E[ r^{(c)}_n(\tau)]$.
\label{Thm:flow-control}
\end{Thm}
\begin{proof} Please refer to Appendix  G.
\end{proof}

\begin{Rem} 
Theorem~\ref{Thm:flow-control} should be interpreted as
follows. When $\mathbf 0 \prec \boldsymbol \epsilon \in
\text{int}(\Lambda)$, one can choose a {\em finite} $\mathbf
z \succ \mathbf 0$ such that $\boldsymbol \epsilon_{\mathbf
z} \preceq
\boldsymbol \epsilon$.
In this case, the limiting average total queue size under Algorithm 2 is
upper bounded as in~\eqref{eqn:enhanced-flow-DBP-U}. At the same time,
the limiting sum utility is lower bounded as
in~\eqref{eqn:enhanced-flow-DBP-g}, thereby specifying a utility-delay
tradeoff. 
In Section~\ref{sec:simu}, we will demonstrate numerically
that Algorithm 2 indeed yields a better utility-delay tradeoff than the
traditional flow control-BP algorithm  in \cite[pp. 90]{Georgiadis-Neely-Tassiulas:2006}.
When $\boldsymbol \epsilon = \mathbf 0$, i.e., no margin is given,
$z_n^{(c)}$ is chosen to be infinity for all $n\in \mathcal N$ and $c\in \mathcal C$ (i.e., $\mathbf f = \mathbf 0$). In this case, Algorithm 2 reduces to
the traditional flow control-BP algorithm, and
Theorem~\ref{Thm:flow-control} reduces to the traditional utility-delay
tradeoff result (Theorem 5.8 in~\cite{Georgiadis-Neely-Tassiulas:2006}).
\end{Rem}

\section{Numerical Experiments}
\label{sec:simu}


In the numerical simulations, we consider the same simulation setup as in \cite{YingShakkottai08oncombSPDBPJNET} for ease of comparison. Specifically, we consider a  network with 64 nodes  and four clusters as shown in Fig.~\ref{Fig:network}. Each cluster is a $4\times4$ regular grid with two randomly added links. Two adjacent clusters are connected by two links. All links are bidirectional with a maximum transmission rate of one packet/slot  for both directions.  We consider the wireline case, in which all links can transmit simultaneously.
As in \cite{YingShakkottai08oncombSPDBPJNET}, we consider 8 commodities corresponding to the following source-destination pairs: $((1,3),(2,5))$, $((2,3),(2,7))$, $((2,2),(1,6))$, $((3,4),(2,7))$, $((1,1),(1,7))$, $((4,3),(5,4))$, $((4,6),(6,6))$ and $((5,3),(5,6))$.  
The packet arrival processes are Poisson. We compare the performance of our enhanced BP algorithms discussed in Section~\ref{sec:ex}, i.e., BPnxt and BPmin, and the shortest path biased versions of our enhanced BP algorithms, i.e., BPnxtbias and BPminbias,\footnote{As in BPbias, we add two QSI-independent shortest path bias terms $B_a$ and $B_b$ to the instantaneous QSI at nodes $a$ and $b$  in  \eqref{eqn:weight}, where $B_a$ and $B_b$ are   parameterized by the per-link cost $B$.} with several baseline schemes, such as the traditional BP algorithm \cite{Tassiulas-Ephremides:1992-2,Georgiadis-Neely-Tassiulas:2006}, the BPbias algorithm  \cite{Neely-Modiano-Rohrs:2005,Georgiadis-Neely-Tassiulas:2006}, and the BPSP algorithm \cite{YingShakkottai08oncombSPDBPJNET}. 
In the simulations, we use the average number of packets in the network   as the performance measure, a quantity which is linearly related to the average delay by Little's Law.  The average performance is evaluated over $10^5$ time slots.

\begin{figure}[t]
\begin{center}
\includegraphics[height=3.5cm]{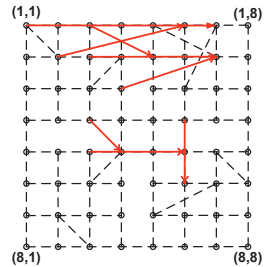}
\caption{\small{Network topology and commodities\cite{YingShakkottai08oncombSPDBPJNET}.  $d_{in}= 5$ and $L=224$. }}\label{Fig:network}
\end{center}
\end{figure}

\begin{figure}[t]
\begin{center}
\includegraphics[height=6cm, width=8.5cm]{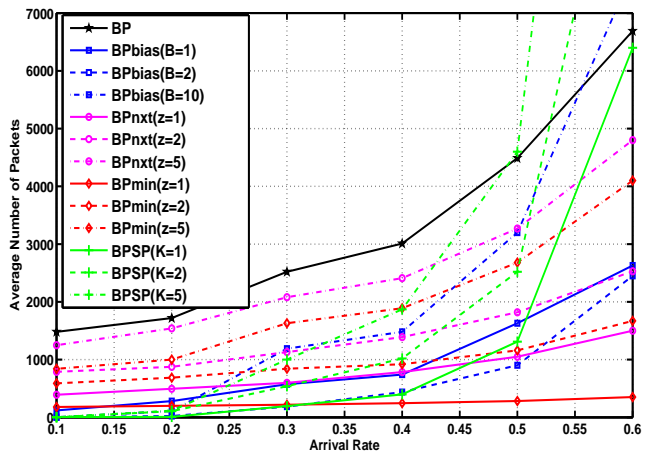}
\caption{\small{Delay of BPnxt  and BPmin at $  z=1,2,5$.}}\label{Fig:proposed baselines}
\end{center}
\end{figure}

\begin{figure}[t]
\begin{center}
\includegraphics[height=6cm, width=8.5cm]{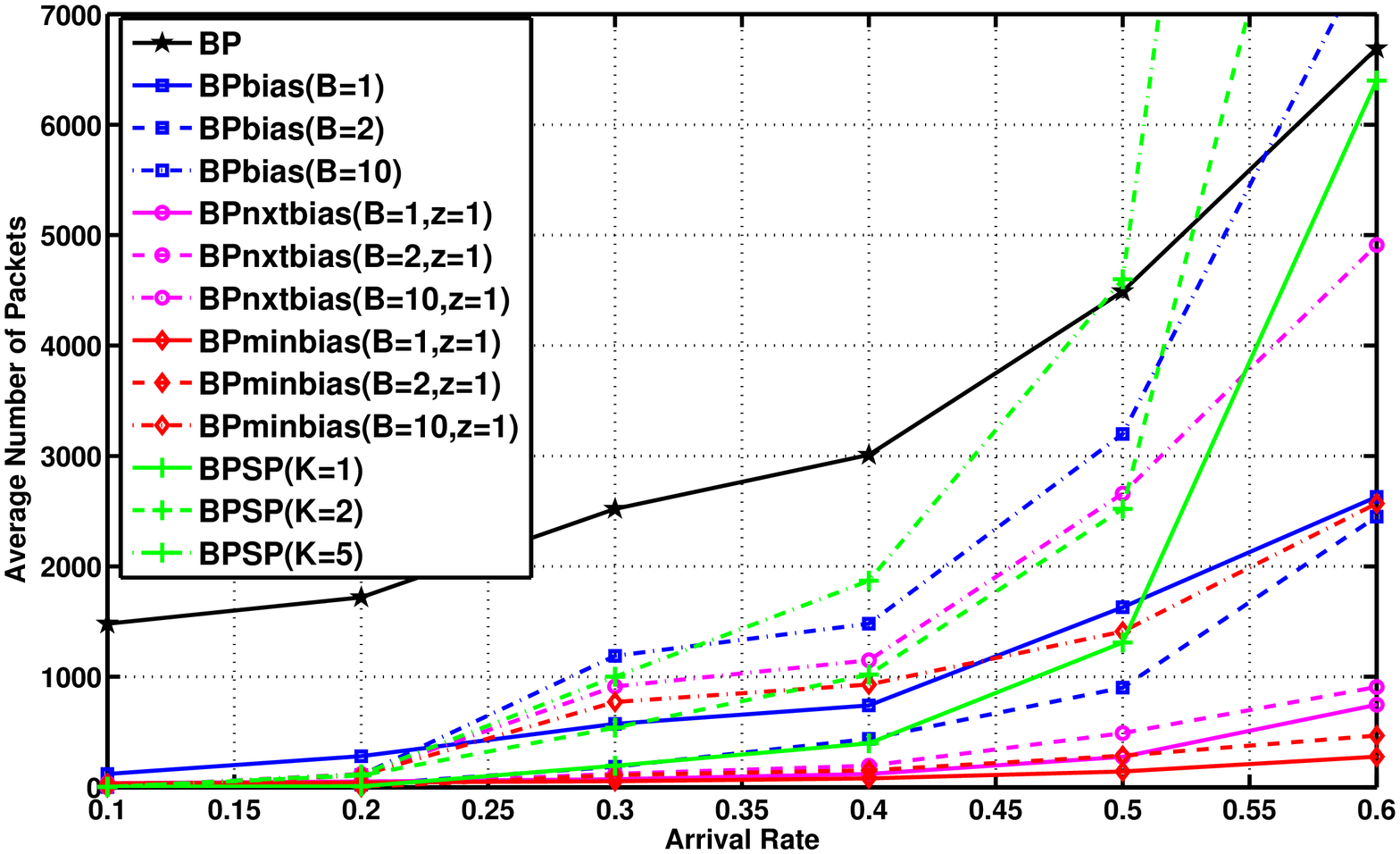}
\caption{\small{Delay   of BPnxtbias  and BPminbias  at $z=1$.}}\label{Fig:proposedbias1}
\end{center}
\end{figure}

\begin{figure}[t]
\begin{center}
\includegraphics[height=6cm, width=8.5cm]{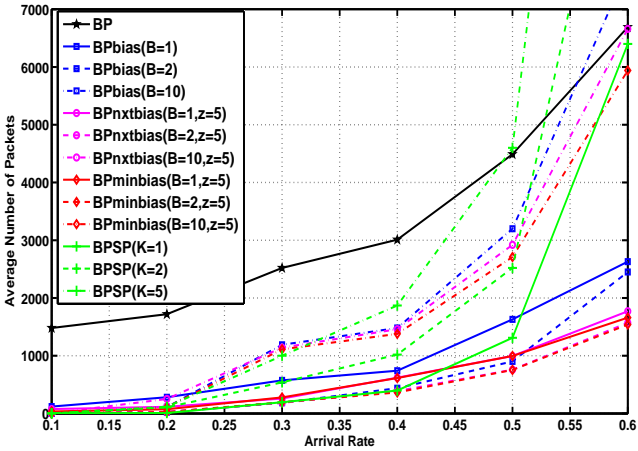}
\caption{\small{Delay  of BPnxtbias  and BPminbias at $z=5$.}}\label{Fig:proposedbias5_224}
\end{center}
\end{figure}

\subsection{Enhanced BP Algorithms}\label{subsec:DBP}

Figures~\ref{Fig:proposed baselines},~\ref{Fig:proposedbias1}  and~\ref{Fig:proposedbias5_224}  show the average number of packets in the
network versus the arrival rate  in the light and moderate loading regimes.  Here, all
commodities have the same arrival rate, i.e., $\lambda^{(c)}_n=\lambda$ for $n=src(c)$ and $c\in \mathcal C$, where $src(c)$ denotes the source node of commodity $c$.  
First, from Fig.~\ref{Fig:proposed baselines}, we can see that   with the minimum next-hop queue length bias and  the minimum downstream sum queue length bias, the delay performance of the traditional BP  can be improved, by using BPnxt  ($  z=1,2,5$) and BPmin  ($z=1,2,5$), respectively. 
It can be verified that arrival rates $\lambda\leq 2.5$ can be stabilized \cite[p. 32]{Georgiadis-Neely-Tassiulas:2006}.  Thus, when $\lambda \leq 0.5$, $z=5$  satisfies the sufficient condition in \eqref{eqn:cond-z-ex} for throughput optimality of BPnxt and BPmin.  On the other hand, as discussed in Section~\ref{Subsec:DBP-performance}, the choice of  $z$  satisfying \eqref{eqn:cond-z-ex} is not necessarily optimal for delay performance.
From Fig.~\ref{Fig:proposed baselines}, it can be seen that the delay for BPnxt (BPmin) with $   z=1$  is at most $28.7\%$ ($12.1\%$) of the delay for BP  for $\lambda=0.1, \cdots, 0.6$. When $   z$  increases,  the delay performance  gain of  BPnxt  (BPmin) over BP decreases, as the effect of the queue length bias reduces.  
In addition,   we see that the delay performance of BPbias and BPSP is very sensitive to the choices of parameters $B$ and $K$, where $B$ is the per-link cost in obtaining the shortest path bias and $K$ is the control parameter for traffic splitting. 
Specifically, when $B$ and $K$ are small, the delay performance improvement of BPbias and BPSP over traditional BP is not significant. When $B$ and $K$ are  large, as compared with traditional BP, BPbias and BPSP have significantly lower delay in the small traffic loading regime.  On the other hand, in the moderate (and heavy) traffic loading regime, a small delay reduction or even a delay increase is observed.  
This is because BPbias and BPSP, by improving delay performance using the shortest path concept, may easily  cause heavy congestion on the shorter paths when $B$ and $K$ are large or when the traffic load is high. Thus, it is difficult in general to determine beforehand the proper $B$ and $K$ parameters. 
In contrast, for small $   z$, the delay of the proposed BPnxt and BPmin algorithms, which improve delay performance by exploiting (dynamic) downstream congestion information, is small across the light and moderate loading regimes.

Next, from Fig.~\ref{Fig:proposedbias1}  and Fig.~\ref{Fig:proposedbias5_224}, we can see that with the additional
minimum next-hop queue length bias and  minimum downstream sum queue length bias, the delay performance of  BPbias ($B=1,2,10$) can be further improved, under the same parameter  $B$.  This supports our conjecture that by considering more QSI, we can substantially improve the delay performance of BP-based algorithms. 
Similar to BPbias, the delay performance of BPnxtbias and BPminbias is also sensitive to the choice of $B$.  However, with  small $B$ and $   z$, BPnxtbias  (BPminbias) can achieve good delay performance across the light and moderate loading regimes.
For example,  the delay of BPnxtbias (BPminbias) at $   z=1$ and  $B=1$ is at most $11.2\%$ ($4.1\%$) of the delay of BP  for $\lambda=0.1, \cdots, 0.6$. 


\subsection{Enhanced Joint Flow Control and BP Algorithms}

\begin{figure}[t]
\begin{center}
\includegraphics[height=6cm, width=8.5cm]{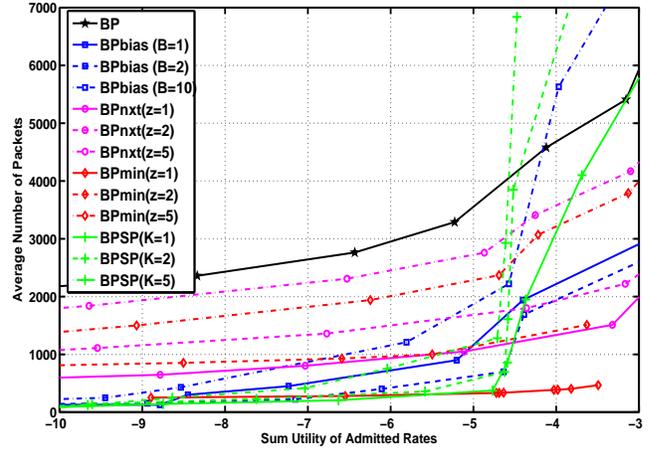} 
\caption{\small{Utility-delay tradeoff of BPnxt   and BPmin at $z=1,2,5$.}} \label{Fig:flow1}
\end{center}
\end{figure}

\begin{figure}[t]
\begin{center}
\includegraphics[height=6cm, width=8.5cm]{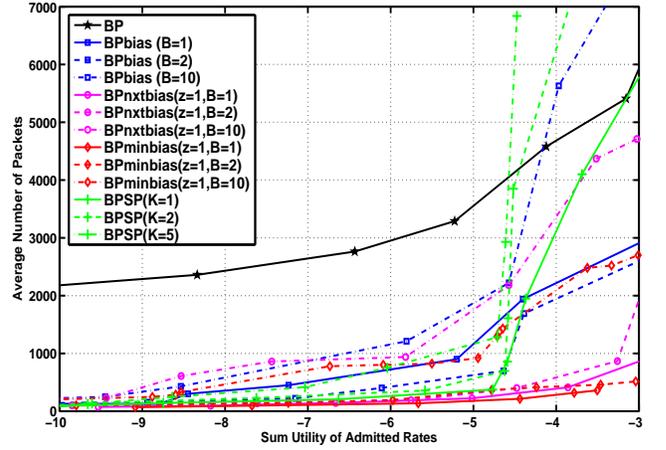}
\caption{\small{Utility-delay tradeoff of BPnxtbias  and BPminbias  at $z=1$.}} \label{Fig:flow2}
\end{center}
\end{figure}

Fig.~\ref{Fig:flow1} and Fig.~\ref{Fig:flow2}  illustrate the average number of packets in the network versus the sum utility of the admitted rate over all commodities.  We consider proportional fairness by choosing the logarithmic utility function. Specifically, for each commodity $c\in \mathcal C$, we choose $h_n^{(c)}(x)=\log(x)$ for $n=src(c)$ and $h_n^{(c)}(x)=0$ for all $n\neq src(c)$.   
We choose  $\lambda^{(c)}_{n}=3$ and $r^{(c)}_{n,\max}=1$ for  $n=src(c)$. The utility-delay tradeoff curve is obtained by choosing different control parameters $M$.  As in Section~\ref{subsec:DBP}, we see that the utility-delay tradeoff for BPnxt and BPmin are significantly improved when compared with that of traditional BP.  In addition, the performance for BPnxtbias and BPminbias are very competitive when compared with BPbias and BPSP.

\subsection{Discussion}

\begin{table}[h]
\begin{center}
\begin{tabular}{|c|c|}
\hline
Algorithm &  Normalized Simulation Time \\
  \hline
  BP  (BPbias) & 1 (1.1) \\
  \hline
  BPnxt (BPnxtbias) & 1.8 (1.9) \\
  \hline
  BPmin  (BPminbias) & 12.6 (12.7) \\
  \hline
  BPSP &110.5\\
  \hline
\end{tabular}
  \caption{\small{Comparisons on average normalized simulation time.}}
  \label{Tab:comparison-time}
\end{center}
\end{table}

Table~\ref{Tab:comparison-time} shows the average normalized simulation time  of all the considered algorithms for the network topology and commodities in Fig.~\ref{Fig:network}.  We can see that the BP, BPbias, BPnxt, BPnxtbias, BPmin, BPminbias and BPSP  algorithms are in increasing order of complexity.  
When only low computation cost is acceptable, BPnxtbias, i.e., the combination of a small shortest path bias (i.e., small $B$)  and the minimum next-hop queue length bias with small parameter (i.e., small $  z$), seems to result in  delay performance close to or better than that of BPbias and BPSP across the light and moderate loading regimes. The small shortest path bias captures essential path length information,  and the minimum next-hop queue length bias with small parameter $z$ captures essential congestion information on each path. Both bias terms help to correct the myopic nature of the traditional BP algorithm. When high computation cost is acceptable, e.g., in small networks, we can consider BPmin with small $z$ or BPminbias  with small $B$ and  $z$  for further delay performance  improvement.

\section{Conclusion}

In this paper,  we show that the asymptotically delay optimal control resembles the BP algorithm in basing resource allocation and routing on a backpressure calculation, but differs from the BP algorithm in the form of the backpressure calculation employed. 
Motivated by this connection, we introduce a new class of enhanced BP-based algorithms which incorporate  a general queue-dependent bias function into the BP backpressure term to substantially improve  delay performance. We  demonstrate the throughput optimality and the utility-delay tradeoff for the proposed algorithms.  
We further elaborate on two specific algorithms within this class,  which have demonstrably improved delay performance while maintaining acceptable implementation complexity.

\section*{Appendix A: Proof of Lemma~\ref{Lem:Bellman}}

First,  define a real valued function
\begin{align}
h(s,\mathbf u)\triangleq\min_{I, \boldsymbol \nu
} \left\{\sum_{\substack{n\in \mathcal N\\c\in \mathcal C}}u^{(c)}_n+\sum_{\mathbf u'\in \mathcal U}P_{(s,\mathbf u),\mathbf u'}(I, \boldsymbol \nu)V(\mathbf u')
\right\}-d.\label{eqn:h-def}
\end{align}
By \eqref{eqn:h-def} and \eqref{eqn:Bellman}, we have 
$d +V(\mathbf u)=\sum_{s\in \mathcal S}\Pr[S=s](h(s,\mathbf u)+d)$, implying: 
\begin{align}
V(\mathbf u)=\sum_{s\in \mathcal S}\Pr[S=s]h(s,\mathbf u). \label{eqn:v-def}
\end{align}
Substituting \eqref{eqn:v-def} into \eqref{eqn:h-def}, we have: 
\begin{align}
&d +h(s,\mathbf u)\nonumber\\
\stackrel{(a)}{=}&\min_{I, \boldsymbol \nu} \left\{\sum_{\substack{n\in \mathcal N\\c\in \mathcal C}}u^{(c)}_n+\sum_{s',\mathbf u'}P_{(s,\mathbf u),(s',\mathbf u')}(I, \boldsymbol \nu)h(s',\mathbf u')
\right\},\label{eqn:Bellman-full}
\end{align}
where (a) is due to \eqref{eqn:tran-prob}. In addition, we have 
$\mathbb E[h(S(t),\mathbf U(t)|(S(0),\mathbf U(0))=(s,\mathbf u),\omega
]
\stackrel{(b)}{=}\sum_{\mathbf u'}\Pr[\mathbf U(t)=\mathbf u'|(S(0),\mathbf U(0))=(s,\mathbf u),\omega
]\times\left(\sum_{s'}\Pr[S(t)=s']h(s',\mathbf u')\right)
\stackrel{(c)}{=}\mathbb E[V(\mathbf U(t))|(S(0),\mathbf U(0))=(s,\mathbf u),\omega
]$, 
where (b) is due to the  i.i.d. property of $S(t)$ and (c) is due to \eqref{eqn:v-def}. 
Thus, by \eqref{eqn:Bellman-bound}, we have: 
\begin{align}
\lim_{t\to \infty}\frac{1}{t}\mathbb E[h(S(t),\mathbf U(t)|(S(0),\mathbf U(0))=(s,\mathbf u),\omega
]=0.\label{eqn:Bellman-full-bound}
\end{align}
Note that \eqref{eqn:Bellman-full} and \eqref{eqn:Bellman-full-bound} correspond to conditions (4.121) and (4.122) of Proposition 4.6.1 in \cite[pp.  254]{Bertsekas:2007}.  
In addition, $\mathcal S\times \mathcal U$ is countably infinite and $\mathcal I\times \mathcal R$ is finite. 
Thus,  by Proposition 4.6.1 \cite[pp. 254]{Bertsekas:2007}, we can prove Lemma~\ref{Lem:Bellman}. 

\section*{Appendix B: Proof of Lemma~\ref{Lem:asym-opt-policy}}

Let $\mathbf u'\triangleq\mathbf U(t+1)$, $\mathbf u \triangleq \mathbf U(t)$, $A_n^{(c)}\triangleq A_n^{(c)}(t)$ and $\nu^{(c)}_{ab}\triangleq \nu^{(c)}_{ab}(t)$. The vector form of \eqref{eqn:queue_dyn-mdp} can be written as
$\mathbf u'=\mathbf u-
\left(\sum_{b\in \mathcal N}\nu^{(c)}_{nb}\Delta\right)+\left(A_n^{(c)}\Delta\right)+\left(\sum_{a\in \mathcal N}\nu^{(c)}_{an}\Delta\right)$. 
 By Taylor's theorem for multivariate functions, we have $V(\mathbf u')= 
V(\mathbf u)+\Delta\sum_{\substack{n\in \mathcal N, c\in \mathcal C}} V_{n,(c)}'(\mathbf u)\left(A_n^{(c)}+\sum_{a\in \mathcal N}\nu^{(c)}_{an}-\sum_{b\in \mathcal N}\nu^{(c)}_{nb}\right)+o(\Delta)$\cite{CuitwohopIT2015}. 
Thus, we have: 
\begin{align}
&\sum_{\mathbf u'\in \mathcal U}P_{(s,\mathbf u),\mathbf u'}(I, \boldsymbol \nu)V(\mathbf u')=V(\mathbf u)+\Delta\sum_{n\in \mathcal N, c\in \mathcal C} V_{n,(c)}'(\mathbf u)\lambda_n^{(c)}\nonumber\\
&-\Delta\sum_{(a,b)\in \mathcal L}\sum_{c\in \mathcal C}\nu_{ab}^{(c)}\left(V_{a,(c)}'(\mathbf u)- V_{b,(c)}'(\mathbf u)\right)+o(\Delta).\label{eqn:proof:Taylor-exp-}
\end{align}
By \eqref{eqn:bellman-opt-policy} and \eqref{eqn:proof:Taylor-exp-}, we have:  
\begin{align}
&\omega^*(s,\mathbf u)=\arg\min_{I, \boldsymbol \nu
}\sum_{\mathbf u'\in \mathcal U}P_{(s,\mathbf u),\mathbf u'}(I, \boldsymbol \nu)V(\mathbf u')=\nonumber\\
&\arg\min_{I, \boldsymbol \nu
}\left(-\Delta\sum_{(a,b)\in \mathcal L}\sum_{c\in \mathcal C}\nu_{ab}^{(c)}\left(V_{a,(c)}'(\mathbf u)- V_{b,(c)}'(\mathbf u)\right)+o(\Delta)\right).\label{eqn:proof-prob-opt}
\end{align}
Let $\omega^*(s,\mathbf u)=(I^*, \boldsymbol \nu^*)$, $\omega^{\dagger}(s,\mathbf u)=(I^{\dagger}, \boldsymbol \nu^{\dagger})$, $T^*_{(s,\mathbf u)}\mathbf h\triangleq\sum_{\substack{n\in \mathcal N, c\in \mathcal C}}u^{(c)}_n+\sum_{s',\mathbf u'}P_{(s,\mathbf u),(s',\mathbf u')}(I^*, \boldsymbol \nu^*)h(s',\mathbf u')$, and $T^{\dagger}_{(s,\mathbf u)}\mathbf h\triangleq\sum_{\substack{n\in \mathcal N,c\in \mathcal C}}u^{(c)}_n+\sum_{s',\mathbf u'}P_{(s,\mathbf u),(s',\mathbf u')}(I^{\dagger}, \boldsymbol \nu^{\dagger})h(s',\mathbf u')$, where $\mathbf h\triangleq\left(h(s,\mathbf u)\right)$. 
By \eqref{eqn:proof:Taylor-exp-} and \eqref{eqn:proof-prob-opt}, we have: 
\footnote{Equality \eqref{eqn:proof-opt-asym-relation}  is due to the following. Let $f_1(x)$ and $f_2(x)$ be two functions of $x$. Let $x^*\triangleq \arg\min_x (f_1(x)+f_2(x))$ and $x^{\dagger}\triangleq\arg\min_x f_1(x)$. Then, we have $f_1(x^{\dagger})+f_2(x^*)\leq f_1(x^*)+f_2(x^*)\leq f_1(x^{\dagger})+f_2(x^{\dagger})$.} 
\begin{align}
&\sum_{\mathbf u'\in \mathcal U}P_{(s,\mathbf u),\mathbf u'}(I^*, \boldsymbol \nu^*)V(\mathbf u')\nonumber\\
=&\sum_{\mathbf u'\in \mathcal U}P_{(s,\mathbf u),\mathbf u'}(I^{\dagger}, \boldsymbol \nu^{\dagger})V(\mathbf u')+o(\Delta)\label{eqn:proof-opt-asym-relation}\\
\stackrel{(a)}
{\Rightarrow} &T^*_{(s,\mathbf u)}\mathbf h=T^{\dagger}_{(s,\mathbf u)}\mathbf h+o(\Delta)\nonumber\\
\stackrel{(b)}
{\Rightarrow}  & d(\Delta)+h(s,\mathbf u)=T^{\dagger}_{(s,\mathbf u)}\mathbf h+o(\Delta)\nonumber\\
\Rightarrow &T^{\dagger}_{(s,\mathbf u)}\mathbf h=d(\Delta)+h(s,\mathbf u)+o(\Delta), \ (s,\mathbf u)\in \mathcal S\times \mathcal U,\label{proof:new-perf-asymp}
\end{align}
where (a) is due to \eqref{eqn:v-def} and \eqref{eqn:tran-prob}, and (b) is due to \eqref{eqn:Bellman-full}. 
As in the proofs of Proposition 4.1.6 in \cite[pp. 191]{Bertsekas:2007} and Proposition 4.6.1 in \cite[pp. 254]{Bertsekas:2007}, from \eqref{proof:new-perf-asymp}, we can show $ d^{\dagger}(\Delta)=d(\Delta)+o (\Delta)$ as $\Delta \to 0$.

\section*{Appendix C: Proof of Lemma~\ref{Lem:asym-opt-policy-GQSI}}

Let $\omega^{\dagger}(s,\mathbf u)=(I^{\dagger}, \boldsymbol \nu^{\dagger})$. By \eqref{eqn:Bellman}, \eqref{eqn:proof:Taylor-exp-} and \eqref{eqn:proof-opt-asym-relation}, for all $\mathbf u\in \mathcal U$, we have:  
\begin{align}
d =\sum_{n\in \mathcal N , c\in \mathcal C}u^{(c)}_n+\Delta\sum_{n\in \mathcal N , c\in \mathcal C} V_{n,(c)}'(\mathbf u)F_{n,(c)}(\mathbf u)+o(\Delta)\label{proof:ODE}
\end{align}
where 
\begin{align}
F_{n,(c)}(\mathbf u)=\lambda_n^{(c)}+\sum_{s\in \mathcal S}\Pr[S=s]\left(\sum_{a\in \mathcal N}\nu^{(c)\dagger}_{an}-\sum_{b\in \mathcal N}\nu^{(c)\dagger}_{nb}\right).\nonumber
\end{align}
Suppose for all $n\in \mathcal N$ and $c\in \mathcal C$,  there exists a function $g_n^{(c)}(u_n^{(c)})$, such that $V_{n,(c)}'(\mathbf u)=g_n^{(c)}(u_n^{(c)})$. By \eqref{eqn:prob-asymp}, we know that $F_{n,(c)}(\mathbf u)$   is still a function of the global QSI $\mathbf u$. In addition, by \eqref{proof:ODE} and $V_{n,(c)}'(\mathbf u)=g_n^{(c)}(u_n^{(c)})$, we have:    
\begin{align}
g_{k,(c')}(u_{k}^{(c')})=&\frac{d-\sum_{\substack{n\in \mathcal N\\c\in \mathcal C}}u^{(c)}_n+o(\Delta)}{\Delta F_{k,(c')}(\mathbf u)}\nonumber\\
&-\frac{\Delta\sum_{\substack{n\in \mathcal N, c\in \mathcal C\\ (n,c)\neq(k,c')}} g_{n,(c)}(u_n^{(c)})F_{n,(c)}(\mathbf u)}{\Delta F_{k,(c')}(\mathbf u)}.\nonumber 
\end{align}
For this equality, the R.H.S.   is a function of  $\mathbf u$ and the L.H.S.   is a function of $u_{k}^{(c')}$. Thus, the above equality cannot hold. By contradiction, we can show  that for  some $n\in \mathcal N$ and $c\in \mathcal C$, there does not exist a function $g_n^{(c)}(u_n^{(c)})$, such that $V_{n,(c)}'(\mathbf u)=g_n^{(c)}(u_n^{(c)})$. 
Thus, we can always find  $\mathbf u\in  \mathcal U$, such that $c^{\dagger}_{ab}(\mathbf u)\neq c^{\ddagger}_{ab}(\mathbf u)$ and $\delta V_{ab}^{\dagger}(\mathbf u)\neq \delta V_{ab}^{\ddagger}(\mathbf u)$ for some $(a,b)\in \mathcal L$, where $c^{\ddagger}_{ab}(\mathbf u)$ and $\delta V_{ab}^{\ddagger}(\mathbf u)$ are defined in a similar way to $c^{\dagger}_{ab}(\mathbf u)$ and $\delta V_{ab}^{\dagger}(\mathbf u)$
but with  $u_n^{(c)}$ instead of $V_{n,(c)}'(\mathbf u)$. Therefore, we can find  some $(s,\mathbf u)\in \mathcal S \times  \mathcal U$ such that $\omega^{\dagger}(s,\mathbf u)\neq \omega^{\ddagger}(s,\mathbf u)$.

\section*{Appendix D: Proof of Lemma~\ref{Lem:asym-DBP}}

Let $\omega^{\ddagger}(s,\mathbf u)=(I^{\ddagger}, \boldsymbol \nu^{\ddagger})$ and $T^{\ddagger}_{(s,\mathbf u)}\mathbf h\triangleq\sum_{\substack{n\in \mathcal N\\c\in \mathcal C}}u^{(c)}_n+\sum_{s',\mathbf u'}P_{(s,\mathbf u),(s',\mathbf u')}(I^{\ddagger}, \boldsymbol \nu^{\ddagger})h(s',\mathbf u')$.  By \eqref{eqn:proof:Taylor-exp-}, \eqref{eqn:prob-asymp} and \eqref{eqn:prob-DBP}, we know $T^{\ddagger}_{(s,\mathbf u)}\mathbf h\geq T^{\dagger}_{(s,\mathbf u)}\mathbf h+o(\Delta)$ for all  $ (s,\mathbf u)\in \mathcal S\times \mathcal U$. In addition, by  Lemma ~\ref{Lem:asym-opt-policy-GQSI}, we know that there exist $\epsilon >0$ and $ (s,\mathbf u)\in \mathcal S\times \mathcal U$,  such that  $T^{\ddagger}_{(s,\mathbf u)}\mathbf h\geq T^{\dagger}_{(s,\mathbf u)}\mathbf h+o(\Delta)+\epsilon$. Combining \eqref{proof:new-perf-asymp}, we have for all  $ (s,\mathbf u)\in \mathcal S\times \mathcal U$,  $T^{\ddagger}_{(s,\mathbf u)}\mathbf h\geq d(\Delta)+h(s,\mathbf u)+o(\Delta)$, and for some $ (s,\mathbf u)\in \mathcal S\times \mathcal U$, $T^{\ddagger}_{(s,\mathbf u)}\mathbf h\geq  d(\Delta)+h(s,\mathbf u)+o(\Delta)+\epsilon$.   As in the proofs of Propositions 4.1.6 in \cite[pp. 191]{Bertsekas:2007} and  4.6.1 in \cite[pp. 254]{Bertsekas:2007}, we can show $d^{\ddagger}(\Delta)\geq d(\Delta)+o(\Delta)+\epsilon$. Thus, by Lemma~\ref{Lem:asym-opt-policy}, we have $d^{\ddagger}(\Delta)-d^{\dagger}(\Delta)\geq \epsilon+o(\Delta)$ as $\Delta \to 0$.

\section*{Appendix E: Proof of Theorem~\ref{Thm:thpt-opt}}

Define the Lyapunov function $L(\mathbf u)\triangleq
\sum_{n\in \mathcal N,c\in \mathcal C}(u^{(c)}_n)^2$. The Lyapunov drift at slot $t$ is
$\Delta (\mathbf U(t))\triangleq \mathbb E[L\big(\mathbf
U(t+1)\big)-L\big(\mathbf U(t)\big)|\mathbf U(t)]$.  Squaring both sides of \eqref{eqn:queue_dyn} and following steps similar to those in \cite{Georgiadis-Neely-Tassiulas:2006}, we have:
\footnote{Note that $\mu^{(c)}_{ab}(t)$ denotes the action of Algorithm~\ref{Alg:enhanced-DBP}.} 
\begin{align}
&L\left(\mathbf U(t+1)\right)-L\left(\mathbf U(t)\right)\nonumber\\
&\leq 2N\bar B+2\sum_{n\in \mathcal N,c\in \mathcal C}U^{(c)}_n(t)A^{(c)}_n(t)-2\sum_{(a,b)\in \mathcal L}\sum_{c\in \mathcal C}
\mu^{(c)}_{ab}(t)\nonumber\\
&\quad \times\left(\left(U^{(c)}_a(t)+f_a^{(c)}(\mathbf U(t))\right)-\left(U^{(c)}_b(t)+f_b^{(c)}(\mathbf U(t))\right)\right)\nonumber\\
&\quad +2\sum_{(a,b)\in \mathcal L}\sum_{c\in \mathcal C}
\mu^{(c)}_{ab}(t)\left( f_a^{(c)}(\mathbf U(t))- f_b^{(c)}(\mathbf U(t))\right)\nonumber\\
& \stackrel{(a)}{\leq} 2N\bar B+2\sum_{n\in \mathcal N,c\in \mathcal C}U^{(c)}_n(t)A^{(c)}_n(t)-2\sum_{(a,b)\in \mathcal L}\sum_{c\in \mathcal C}
\mu^{(c)}_{ab}(t)\nonumber\\
& \quad\times\left(\left(U^{(c)}_a(t)+f_a^{(c)}(\mathbf U(t))\right)-\left(U^{(c)}_b(t)+f_b^{(c)}(\mathbf U(t))\right)\right)\nonumber\\
& \quad +2\sum_{n\in \mathcal N,c\in \mathcal C}U^{(c)}_n(t)\frac{R_{\max}L^{(c)}}{z^{(c)}_n}\label{eqn:proof-deltaL}
\end{align}
where (a) is due to the following:
\begin{align}
&\sum_{(a,b)\in \mathcal L}\sum_{c\in \mathcal C}
\mu^{(c)}_{ab}(t)\left( f_a^{(c)}(\mathbf U(t))- f_b^{(c)}(\mathbf U(t))\right)\nonumber\\
&= \sum_{(a,b)\in \mathcal L}\sum_{c\in \mathcal C}\mu^{(c)}_{ab}(t)\sum_{k\in \mathcal N}\frac{\eta_{ak}^{(c)}(\mathbf U(t))-\eta_{bk}^{(c)}(\mathbf U(t))}{z_{k}^{(c)}} U_{k}^{(c)}(t)\nonumber\\
&\leq \sum_{c\in \mathcal C}\sum_{(a,b)\in \mathcal L^{(c)}} R_{\max}\sum_{k\in \mathcal N}\frac{1}{z_{k}^{(c)}} U_{k}^{(c)}(t)\nonumber\\
&=\sum_{n\in \mathcal N,c\in \mathcal C}U^{(c)}_n(t)\frac{R_{\max}L^{(c)}}{z^{(c)}_n}\label{eqn:proof-bound-fab}.
\end{align}
Taking conditional
expectations on both sides of \eqref{eqn:proof-deltaL}, we have: 
\begin{align}
&\Delta \left(\mathbf
U(t)\right) \nonumber\\
&\stackrel{(b)}{\leq}
2N\bar B+2\sum_{n\in \mathcal N,c\in \mathcal C}U^{(c)}_n(t)\lambda^{(c)}_n-2\mathbb
E\Bigg[\sum_{(a,b)\in \mathcal L}\sum_{c\in \mathcal C}\tilde{\mu}^{(c)}_{ab}(t)
\nonumber\\
& \times\left(\left(U^{(c)}_a(t)+f_a^{(c)}(\mathbf U(t))\right)-\left(U^{(c)}_b(t)+f_b^{(c)}(\mathbf U(t))\right)\right)\Big|\mathbf U(t)\Bigg]\nonumber\\
&+2\sum_{n\in \mathcal N,c\in \mathcal C}U^{(c)}_n(t)\frac{R_{\max}L^{(c)}}{z^{(c)}_n}\nonumber\\
&=
2N\bar B+2\sum_{n\in \mathcal N,c\in \mathcal C}U^{(c)}_n(t)\lambda^{(c)}_n\nonumber\\
&-2\sum_{n\in \mathcal N,c\in \mathcal C}U^{(c)}_n(t)\mathbb E\left[\left(\sum_{b\in \mathcal N}\tilde{\mu}^{(c)}_{nb}(t)-\sum_{a\in \mathcal N}\tilde{\mu}^{(c)}_{an}(t)\right)\Big|\mathbf
U(t)\right]\nonumber\\
&+2\sum_{(a,b)\in \mathcal L}\sum_{c\in \mathcal C} \mathbb E\left[\tilde{\mu}^{(c)}_{ab}(t) \left( f_b^{(c)}(\mathbf U(t))-f_a^{(c)}(\mathbf U(t))\right)\Big|\mathbf U(t)\right]\nonumber\\
&+2\sum_{n\in \mathcal N,c\in \mathcal C}U^{(c)}_n(t)\frac{R_{\max}L^{(c)}}{z^{(c)}_n}\label{eqn:proof_ineq0}
\end{align}
where (b) is due to the fact that Algorithm~\ref{Alg:enhanced-DBP} minimizes the
R.H.S. of  (b) over all possible alternative
actions  $\tilde{\mu}^{(c)}_{ab}(t)$. 
Since $\boldsymbol
\lambda+\boldsymbol \epsilon+\boldsymbol \delta \in \Lambda$, by Corollary 3.9 of \cite{Georgiadis-Neely-Tassiulas:2006}, there exists a stationary
randomized policy that makes decisions based only on $S(t)$ (i.e.
independent of $\mathbf U(t)$) such that
\begin{align}
\mathbb E\left[\left(\sum_{b\in \mathcal N}\tilde{\mu}^{(c)}_{nb}(t)-\sum_{a\in \mathcal N}\tilde{\mu}^{(c)}_{an}(t)\right)\Big|\mathbf
U(t)\right]=\epsilon^{(c)}_n+\delta_n^{(c)}+\lambda^{(c)}_n\label{eqn:proof_ineq1}.
\end{align}
On the other hand, similar to \eqref{eqn:proof-bound-fab}, we have: 
\begin{align}
&\mathbb  E\left[\sum_{(a,b)\in \mathcal L}\sum_{c\in \mathcal C} \tilde{\mu}^{(c)}_{ab}(t) \left( f_b^{(c)}(\mathbf U(t))-f_a^{(c)}(\mathbf U(t))\right)\Big|\mathbf U(t)\right]\nonumber\\
& \leq \sum_{n\in \mathcal N,c\in \mathcal C}U^{(c)}_n(t)\frac{R_{\max}L^{(c)}}{z^{(c)}_n}\label{eqn:proof_ineq2}.
\end{align}
Substituting \eqref{eqn:proof_ineq1} and \eqref{eqn:proof_ineq2}
into \eqref{eqn:proof_ineq0}, we have $\Delta (\mathbf U(t))
\leq
2N\bar B-2\min_{ n\in \mathcal
N,c\in \mathcal C}\left\{\epsilon_n^{(c)}+\delta_n^{(c)}
-\frac{2R_{\max}L^{(c)}}{z_n^{(c)}}\right\}\times\sum_{n\in \mathcal N,c\in \mathcal C}U^{(c)}_n(t) $. 
By  Lemma 4.1 of \cite{Georgiadis-Neely-Tassiulas:2006} and by minimizing the upper bound over all possible $(\boldsymbol \epsilon,\boldsymbol \delta )$, we complete the proof of Theorem~\ref{Thm:thpt-opt}.

\section*{Appendix F: Proof of Sufficient Condition for $z$}

Using the proof of Theorem~\ref{Thm:thpt-opt}, we replace \eqref{eqn:proof-bound-fab} with \eqref{eqn:proof-bound-fab-nxt} to show that for all $z$ satisfying \eqref{eqn:cond-z-ex},
BPnxt stabilizes the network for any $\boldsymbol \lambda$ satisfying $\boldsymbol \lambda+\boldsymbol \epsilon \in \text{int}(\Lambda)$:
\begin{align}
&\sum_{(a,b)\in \mathcal L}\sum_{c\in \mathcal C}
\mu^{(c)}_{ab}(t)\left( f_a^{(c)}\left(\mathbf U_a^{(c)}(t)\right)- f_b^{(c)}\left(\mathbf U_b^{(c)}(t)\right)\right)\nonumber\\
& \leq\sum_{(a,b)\in \mathcal L}\sum_{c\in \mathcal C}
\mu^{(c)}_{ab}(t)\frac{H_a^{*(c)}\left(\mathbf U_a^{(c)}(t)\right)}{z}\nonumber\\
& \leq \sum_{(a,b)\in \mathcal L}\sum_{c\in \mathcal C}\mu^{(c)}_{ab}(t)\frac{\sum_{k\in \mathcal N_{out,a}^{(c)}} U_{k}^{(c)}(t)}{N_{out,a}^{(c)}z}\nonumber\\
& \leq \sum_{c\in \mathcal C}\sum_{(a,b)\in \mathcal L^{(c)}} R_{\max}\frac{\sum_{k\in \mathcal N_{out,a}^{(c)}} U_{k}^{(c)}(t)}{N_{out,a}^{(c)}z}\nonumber\\
&= \sum_{c\in \mathcal C}\sum_{a\in\mathcal N}\sum_{b\in \mathcal N_{out,a}^{(c)} }R_{\max}\frac{\sum_{k\in \mathcal N_{out,a}^{(c)}} U_{k}^{(c)}(t)}{N_{out,a}^{(c)}z}\nonumber\\
&=\sum_{c\in \mathcal C}\sum_{a\in\mathcal N}\sum_{k\in \mathcal N_{out,a}^{(c)}} U_{k}^{(c)}(t)\frac{R_{\max}}{z}\nonumber\\
& =\sum_{n\in \mathcal N,c\in \mathcal C}U^{(c)}_n(t)\frac{R_{\max} N_{in,n}^{(c)}}{z}\nonumber\\
&\leq\sum_{n\in \mathcal N,c\in \mathcal C}U^{(c)}_n(t)\frac{R_{\max}d_{in}}{z}.\label{eqn:proof-bound-fab-nxt}
\end{align}Here, 
$\mathcal N_{out,n}^{(c)}\triangleq \{k:(n,k)\in\mathcal L^{(c)}\}$, $ N_{out,n}^{(c)}\triangleq|\mathcal N_{out,n}^{(c)}|$, $\mathcal N_{in,n}^{(c)}\triangleq \{k:(k,n)\in\mathcal L^{(c)}\}$, $ N_{in,n}^{(c)}\triangleq|\mathcal N_{in,n}^{(c)}|$, and $d_{in}=\max_{n\in \mathcal N, c\in \mathcal C}N_{in,n}^{(c)}$. 

Similarly, we replace \eqref{eqn:proof-bound-fab} with \eqref{eqn:proof-bound-fab-min} to show that for all $z$ satisfying \eqref{eqn:cond-z-ex},
BPmin stabilizes the network for any $\boldsymbol \lambda$ satisfying $\boldsymbol \lambda+\boldsymbol \epsilon \in \text{int}(\Lambda)$:
\begin{align}
&\sum_{(a,b)\in \mathcal L}\sum_{c\in \mathcal C}
\mu^{(c)}_{ab}(t)\left( f_a^{(c)}\left(\mathbf U(t)\right)- f_b^{(c)}\left(\mathbf U(t)\right)\right)\nonumber\\
&  \leq\sum_{(a,b)\in \mathcal L}\sum_{c\in \mathcal C}
\mu^{(c)}_{ab}(t)\frac{T_b^{*(c)}(\mathbf U(t))+U_{b}^{(c)}(t)-T_b^{*(c)}(\mathbf U(t))}{z}\nonumber\\
& = \sum_{(a,b)\in \mathcal L}\sum_{c\in \mathcal C}\mu^{(c)}_{ab}(t)\frac{ U_{b}^{(c)}(t)}{z}\nonumber\\
&
{\leq} 
\sum_{c\in \mathcal C}\sum_{(a,b)\in \mathcal L^{(c)}} R_{\max}\frac{  U_{b}^{(c)}(t)}{z}\nonumber\\
&=\sum_{c\in \mathcal C}\sum_{a\in\mathcal N}\sum_{b\in \mathcal N_{out,a}^{(c)} }R_{\max}\frac{ U_{b}^{(c)}(t)}{z}\nonumber\\
& =\sum_{n\in \mathcal N,c\in \mathcal C}U^{(c)}_n(t)\frac{R_{\max} N_{in,n}^{(c)}}{z}\nonumber\\
&\leq\sum_{n\in \mathcal N,c\in \mathcal C}U^{(c)}_n(t)\frac{R_{\max}d_{in}}{z}.\label{eqn:proof-bound-fab-min}
\end{align}

\section*{Appendix G: Proof of Theorem~\ref{Thm:flow-control}}

Define the Lyapunov function $L(\boldsymbol \theta)\triangleq\frac{1}{2}
\sum_{n\in \mathcal N,c\in \mathcal C}\left((u^{(c)}_n)^2+(y^{(c)}_n)^2\right)$, where $\boldsymbol \theta\triangleq (\mathbf u, \mathbf y)$. Denote $\boldsymbol \Theta(t)\triangleq (\mathbf U(t), \mathbf Y(t))$.  The Lyapunov drift at slot $t$ is
$\Delta (\boldsymbol \Theta(t))\triangleq \mathbb E[L\big(\boldsymbol \Theta(t+1)\big)-L\left(\boldsymbol \Theta(t)\right)|\boldsymbol \Theta(t)]$.  
Squaring both sides of \eqref{eqn:queue_dyn-flow} and \eqref{eqn:queue_dyn-virtual} and following steps similar to those in \cite{Georgiadis-Neely-Tassiulas:2006}, we have:
\footnote{Note that $r^{(c)}_n (t)$, $\gamma^{(c)}_n(t)$ and $\mu^{(c)}_{ab}(t)$ 
denote the actions of Algorithm~\ref{Alg:enhanced-flow-DBP}.} 
\begin{align}
&L\left(\boldsymbol \Theta (t+1)\right)-L\left(\boldsymbol \Theta (t)\right) \nonumber\\
&\leq   N\hat B+ \sum_{n\in \mathcal N,c\in \mathcal C}U^{(c)}_n(t)r^{(c)}_n(t)- \sum_{(a,b)\in \mathcal L}\sum_{c\in \mathcal C}
\mu^{(c)}_{ab}(t)\nonumber\\
&\quad \times\left(\left(U^{(c)}_a(t)+f_a^{(c)}(\mathbf U(t))\right)-\left(U^{(c)}_b(t)+f_b^{(c)}(\mathbf U(t))\right)\right)\nonumber\\
&\quad +  \sum_{(a,b)\in \mathcal L}\sum_{c\in \mathcal C}
\mu^{(c)}_{ab}(t)\left( f_a^{(c)}(\mathbf U(t))- f_b^{(c)}(\mathbf U(t))\right)\nonumber\\
& \quad - \sum_{n\in \mathcal N,c\in \mathcal C}Y_n^{(c)}(t)\left(r_n^{(c)}(t)-\gamma_n^{(c)}(t)\right)\nonumber\\
& \stackrel{(a)}{\leq} N\hat B+ \sum_{n\in \mathcal N,c\in \mathcal C}U^{(c)}_n(t)r^{(c)}_n(t)- \sum_{(a,b)\in \mathcal L}\sum_{c\in \mathcal C}
\mu^{(c)}_{ab}(t)\nonumber\\
&\quad \times\left(\left(U^{(c)}_a(t)+f_a^{(c)}(\mathbf U(t))\right)-\left(U^{(c)}_b(t)+f_b^{(c)}(\mathbf U(t))\right)\right)\nonumber\\
&\quad +  \sum_{n\in \mathcal N,c\in \mathcal C}U^{(c)}_n(t)\frac{R_{\max}L^{(c)}}{z^{(c)}_n}\nonumber\\
& \quad - \sum_{n\in \mathcal N,c\in \mathcal C}Y_n^{(c)}(t)\left(r_n^{(c)}(t)-\gamma_n^{(c)}(t)\right)\label{eqn:proof-drift-flow}
\end{align}
where (a) is due to \eqref{eqn:proof-bound-fab}.  Taking conditional
expectations and subtracting $M\mathbb E\left[\sum_{n\in \mathcal N,c\in \mathcal C}h^{(c)}_n\left( \gamma^{(c)}_n(t)\right)\Big|\boldsymbol \Theta (t)\right]$ from both sides of \eqref{eqn:proof-drift-flow}, we have: 
\begin{align}
&\Delta \left(\boldsymbol \Theta (t)\right )-M\mathbb E\left[\sum_{n\in \mathcal N,c\in \mathcal C}h^{(c)}_n\left( \gamma^{(c)}_n(t)\right)\Big|\boldsymbol \Theta(t)\right]\nonumber\\
& \stackrel{(b)}{\leq}  N\hat B- \sum_{n\in \mathcal N,c\in \mathcal C}\left(Y^{(c)}_n(t)-U^{(c)}_n(t)\right)\mathbb E\left[\tilde r^{(c)}_n(t)\Big |\boldsymbol \Theta(t)\right]\nonumber\\
& \quad -\sum_{n\in \mathcal N,c\in \mathcal C}\mathbb E\left[Mh^{(c)}_n\left( \tilde \gamma^{(c)}_n(t)\right)- Y_n^{(c)}(t)\tilde \gamma^{(c)}_n(t)\Big|\boldsymbol \Theta(t)\right]\nonumber\\
& \quad -  \sum_{n\in \mathcal N,c\in \mathcal C}U^{(c)}_n(t)\mathbb E\left[\left(\sum_{b\in \mathcal N}\tilde{\mu}^{(c)}_{nb}(t)-\sum_{a\in \mathcal N}\tilde{\mu}^{(c)}_{an}(t)\right)\Big|\boldsymbol \Theta(t)\right]\nonumber\\
& \quad + \sum_{(a,b)\in \mathcal L}\sum_{c\in \mathcal C} \mathbb E\left[\tilde{\mu}^{(c)}_{ab}(t) \left( f_b^{(c)}(\mathbf U(t))-f_a^{(c)}(\mathbf U(t))\right)\Big|\boldsymbol \Theta(t)\right]\nonumber\\
& \quad 
+  \sum_{n\in \mathcal N,c\in \mathcal C}U^{(c)}_n(t)\frac{R_{\max}L^{(c)}}{z^{(c)}_n}
\label{eqn:proof_ineq0-flow}
\end{align}
where (b) is due to the fact that  Algorithm~\ref{Alg:enhanced-flow-DBP} minimizes the
R.H.S. of  (b) over all possible alternative
$\tilde r^{(c)}_n (t)$, $\tilde \gamma^{(c)}_n(t)$ and $\tilde \mu^{(c)}_{ab}(t)$. 
It is not difficult to construct alternative random policies that choose $\tilde r^{(c)}_n (t)$, $\tilde \gamma^{(c)}_n(t)$, $\tilde \mu^{(c)}_{ab}(t)$ such that
\begin{align}
&\mathbb E\left[\tilde r^{(c)}_n(t)\Big|\boldsymbol \Theta(t)\right]=\overline r^{*(c)}_n(\boldsymbol \epsilon+\boldsymbol \delta)\label{eqn:rand-r}\\
& \tilde \gamma^{(c)}_n(t)=\overline r^{*(c)}_n(\boldsymbol \epsilon+\boldsymbol \delta)\label{eqn:rand-gamma}\\
& \mathbb E\left[\left(\sum_{b\in \mathcal N}\tilde{\mu}^{(c)}_{nb}(t)-\sum_{a\in \mathcal N}\tilde{\mu}^{(c)}_{an}(t)\right)\Big|\boldsymbol \Theta(t)\right]\nonumber\\
& = \overline r^{*(c)}_n(\boldsymbol \epsilon+\boldsymbol \delta)+\epsilon^{(c)}_n+\delta_n^{(c)}\label{eqn:rand-mu}
\end{align}
where  $ \overline{\mathbf r}^*(\boldsymbol \epsilon+\boldsymbol \delta)=\left(\overline r^{*(c)}_n(\boldsymbol \epsilon+\boldsymbol \delta)\right)$ is the target $(\boldsymbol \epsilon+\boldsymbol \delta)$-optimal admitted rate given by
\eqref{eqn:eps-opt-prob}.
\footnote{Specifically, \eqref{eqn:rand-r} can be achieved by the randomized policy which sets $\tilde r^{(c)}_n(t)=A^{(c)}_n(t)$ with probability $\overline r^{*(c)}_n(\boldsymbol \epsilon+\boldsymbol \delta)/\lambda^{(c)}_n$ and $\tilde r^{(c)}_n(t)=0$ with probability $1-\overline r^{*(c)}_n(\boldsymbol \epsilon+\boldsymbol \delta)/\lambda^{(c)}_n$.} 
Equation \eqref{eqn:rand-mu} follows from the same arguments leading to~\eqref{eqn:proof_ineq1}.
Thus, by \eqref{eqn:rand-r}, \eqref{eqn:rand-gamma}, \eqref{eqn:rand-mu} and \eqref{eqn:proof_ineq2}, from \eqref{eqn:proof_ineq0-flow}, we obtain $\Delta (\boldsymbol \Theta (t))-M\mathbb E\left[\sum_{n\in \mathcal N,c\in \mathcal C}h^{(c)}_n\left( \gamma^{(c)}_n(t)\right)\Big|\boldsymbol \Theta(t)\right]\leq   N\hat B- \min_{n\in \mathcal N,c\in \mathcal C}\left\{\epsilon_n^{(c)}+\delta_n^{(n)}-\frac{2R_{\max}L^{(c)}}{z_n^{(c)}}\right\} \sum_{n\in \mathcal N,c\in \mathcal C}U_n^{(c)}(t)-M\sum_{n\in \mathcal N,c\in \mathcal C}h^{(c)}_n\left(\overline r^{*(c)}_n(\boldsymbol \epsilon+\boldsymbol \delta)\right)$. 
Applying Theorem 5.4 of \cite{Georgiadis-Neely-Tassiulas:2006}, we have: 
\begin{align}
&\limsup_{t\to\infty}\frac{1}{t}\sum_{\tau=0}^{t-1}\sum_{n\in \mathcal N,c\in \mathcal C} \mathbb
E[U^{(c)}_n(\tau)]\nonumber\\
& \leq \frac{ N\hat B+MH_{\max}}{
 \min_{n\in \mathcal N,c\in \mathcal C}\left\{\epsilon_n^{(c)}+\delta_n^{(c)}-\frac{2R_{\max}L^{(c)}}{z_n^{(c)}}\right\}}\label{eqn:proof-enhanced-flow-DBP-U}\\
&\liminf_{t\to\infty}\sum_{n\in \mathcal N,c\in \mathcal C} h^{(c)}_n\left(\overline \gamma^{(c)}_n(t)\right)\nonumber\\
& \geq
\sum_{n\in \mathcal N,c\in \mathcal C}
h^{(c)}_n\left( \overline r^{*(c)}_n\left(\boldsymbol \epsilon+\boldsymbol \delta\right)\right)-\frac{ N\hat B}{M}\label{eqn:proof-enhanced-flow-DBP-g}
\end{align}
where $\overline \gamma^{(c)}_n(t)\triangleq \frac{1}{t}\sum_{\tau=0}^{t-1}\mathbb E[ \gamma^{(c)}_n(\tau)]$. It is easy to prove $\overline \gamma^{(c)}_n(t)\leq \overline r^{(c)}_n(t)$ by showing the stability of the virtual queues. 
As in \cite[pp. 88]{Georgiadis-Neely-Tassiulas:2006}, we optimize the  R.H.S.s of  \eqref{eqn:proof-enhanced-flow-DBP-U} and \eqref{eqn:proof-enhanced-flow-DBP-g} over all possible $(\boldsymbol \epsilon,\boldsymbol \delta)$. Thus, we can show \eqref{eqn:enhanced-flow-DBP-U} and  \eqref{eqn:enhanced-flow-DBP-g}.

\end{document}